\documentclass[11pt]{amsart}
\usepackage[utf8]{inputenc} 	
\usepackage[T1]{fontenc}
\usepackage{lmodern}
\usepackage{yhmath}

\usepackage{amsthm,amsmath,amssymb,amsbsy,bbm,mathrsfs,supertabular,
eurosym,graphicx,enumitem,xcolor}
\usepackage{bm}
\usepackage{mathtools}

\newtheoremstyle{BBstyle0}  {}{}{\itshape}{}{\bfseries}{}{6pt}{}
\newtheoremstyle{BBstyle1}  {3pt}{3pt}{\rmfamily}{}{\itshape}{: }{3pt}{}
\newtheoremstyle{BBstyle2}  {3pt}{3pt}{\itshape}{}{\bfseries\large}{}{0pt}{}
\newtheoremstyle{BBstyle3}  {}{}{\itshape}{}{\bfseries}{: }{3pt}{}
\newtheoremstyle{BBstyle4}  {}{}{\rmfamily}{}{\bfseries}{}{6pt}{}

\usepackage[normalem]{ulem} 
\usepackage[authoryear]{natbib}
\newtheorem{thm}{Theorem}
\newtheorem{lem}{Lemma}
\newtheorem{prop}{Proposition}
\newtheorem{df}{Definition}
\newtheorem{cor}{Corollary}
\newtheorem{ass}{Assumption}

\theoremstyle{definition}

\newtheorem{rmk}{Remark}
\usepackage[english]{babel}

\usepackage{hyperref}  		     

\newcommand{\pa}[1]{\left({#1}\right)}

\newcommand{\cro}[1]{\left[{#1}\right]}

\newcommand{\ac}[1]{\left\{{#1}\right\}}

\newcommand{\argmin}{\mathop{\rm argmin}}

\newcommand{\E}{{\mathbb{E}}}
 
\renewcommand{\L}{{\mathbb{L}}}
\newcommand{\N}{{\mathbb{N}}}
\renewcommand{\P}{{\mathbb{P}}}
\newcommand{\Q}{{\mathbb{Q}}} 
\newcommand{\R}{{\mathbb{R}}}

\newcommand{\sA}{{\mathscr{A}}}

\newcommand{\sC}{{\mathscr{C}}}

\newcommand{\sE}{{\mathscr{E}}}

\newcommand{\sQ}{{\mathscr{Q}}} 

\newcommand{\sT}{{\mathscr{T}}}

\newcommand{\sW}{{\mathscr{W}}}
\newcommand{\sX}{{\mathscr{X}}}
\newcommand{\sY}{{\mathscr{Y}}}

\DeclareMathAlphabet{\mathscrbf}{OMS}{mdugm}{b}{n}

\newcommand{\sbP}{{\mathscrbf{P}}}
\newcommand{\sbQ}{{\mathscrbf{Q}}}

\newcommand{\cA}{{\mathcal{A}}}

\newcommand{\cC}{{\mathcal{C}}}

\newcommand{\cE}{{\mathcal{E}}}
\newcommand{\cF}{{\mathcal{F}}}
 
\newcommand{\cH}{{\mathcal{H}}}
\newcommand{\cI}{{\mathcal{I}}}

\newcommand{\cK}{{\mathcal{K}}}
\newcommand{\cL}{{\mathcal{L}}} 

\newcommand{\cN}{{\mathcal{N}}}

\newcommand{\cP}{{\mathcal{P}}}

\newcommand{\cS}{{\mathcal{S}}} 
\newcommand{\cT}{{\mathcal{T}}}
\newcommand{\cU}{{\mathcal{U}}}

\newcommand{\cW}{{\mathcal{W}}}
\newcommand{\cX}{{\mathcal{X}}}
\newcommand{\cY}{{\mathcal{Y}}}

\newcommand{\gA}{{\mathbf{A}}}

\newcommand{\gP}{{\mathbf{P}}}
\newcommand{\gQ}{{\mathbf{Q}}} 
\newcommand{\gR}{{\mathbf{R}}}
 
\newcommand{\gT}{{\mathbf{T}}}

\newcommand{\gZ}{{\mathbf{Z}}}

\newcommand{\bs}[1]{\boldsymbol{#1}}

\newcommand{\bsX}{{\bs{X}}}

\newcommand{\ggamma}{\bs{\gamma}}

\newlist{listi}{enumerate}{1}
\setlist[listi,1]{label=(\roman*),ref=(\roman*),align=left}
\newcommand{\eref}[1]{(\ref{#1})}

\newcommand{\1}{1\hskip-2.6pt{\rm l}}

\newcommand{\etc}[1]{#1_1,\ldots,#1_n}

\newcommand{\on}{^{\otimes n}}

\newcommand{\et}{^{\star}}

\newcommand{\eps}{{\varepsilon}}

\usepackage{booktabs}
\usepackage{footnote}
\usepackage{enumitem}
\usepackage{algorithm}
\usepackage{algorithmic}

\def\bsg{{\ggamma}}

\begin{document}
\title[Robust nonparametric regression using deep networks]{Robust nonparametric regression based on deep ReLU neural networks}
\author{Juntong CHEN}
\address{\parbox{\linewidth}{Department of Mathematics (DMATH),\\
University of Luxembourg\\
Maison du nombre\\
6 avenue de la Fonte\\
L-4364 Esch-sur-Alzette\\
Grand Duchy of Luxembourg}}
\email{\vspace{5pt}juntong.chen@uni.lu}
\keywords{Nonparametric regression, robust estimation, deep neural networks, circumventing the curse of dimensionality, supremum of an empirical process}
\subjclass[2010]{Primary 62G35, 62G05; Secondary 68T01}
\thanks{This project has received funding from the European Union's Horizon 2020 research and innovation programme under grant agreement N\textsuperscript{o} 811017}
\date{\today}
\begin{abstract}
In this paper, we consider robust nonparametric regression using deep neural networks with ReLU activation function. While several existing theoretically justified methods are geared towards robustness against identical heavy-tailed noise distributions, the rise of adversarial attacks has emphasized the importance of safeguarding estimation procedures against systematic contamination. We approach this statistical issue by shifting our focus towards estimating conditional distributions. To address it robustly, we introduce a novel estimation procedure based on $\ell$-estimation. Under a mild model assumption, we establish general non-asymptotic risk bounds for the resulting estimators, showcasing their robustness against contamination, outliers, and model misspecification. We then delve into the application of our approach using deep ReLU neural networks. When the model is well-specified and the regression function belongs to an $\alpha$-H\"older class, employing $\ell$-type estimation on suitable networks enables the resulting estimators to achieve the minimax optimal rate of convergence. Additionally, we demonstrate that deep $\ell$-type estimators can circumvent the curse of dimensionality by assuming the regression function closely resembles the composition of several H\"older functions. To attain this, new deep fully-connected ReLU neural networks have been designed to approximate this composition class. This approximation result can be of independent interest.

\end{abstract}
\maketitle

\section{Introduction}\label{intro}
A standard nonparametric regression model takes the form
\begin{equation*}
Y_{i} = f\et(W_{i})+\sigma\varepsilon_{i},\quad i=1,\ldots,n,
\end{equation*}
where $Y_{1},\ldots,Y_{n}$ are real-valued observations, $W_{1},\ldots,W_{n}$ are fixed or random design points in $\sW$ (typically $\sW\subset\R^{d}$ for some positive integer $d$), $\sigma$ is a known positive constant, $\varepsilon_{1},\ldots,\varepsilon_{n}$ are unobserved i.i.d. standard real-valued Gaussian random variables which are independent of $W_{1},\ldots,W_{n}$, and $f\et:\sW\rightarrow\R$ is an unknown regression function that we want to estimate.

A substantial body of literature addresses this problem through the minimization of empirical least squares loss functions. By integrating such a classical estimation approach with various approximation models, several methods have been developed and investigated. These include kernel regression (e.g., \cite{Nadaraya1964} and \cite{Watson1964}), local polynomial regression (e.g., \citeauthor{Fan1992} (\citeyear{Fan1992}, \citeyear{Fan1993})), spline-based regression (e.g., \cite{Wahba1990} and \cite{Friedman1991}), and wavelet-based regression (e.g., \cite{2345967} and \cite{1024691081}), among others. In-depth discussions on different methods and theories related to nonparametric regression can also be found in books such as \cite{Gyorfi2002} and \cite{Tsybakov2009}. Particularly, when $f\et:\cro{0,1}^{d}\rightarrow\R$ is of $\alpha$-smoothness, \cite{Stone1982} demonstrated that the minimax optimal convergence rate is of the order $n^{-2\alpha/(2\alpha+d)}$ with respect to some squared $\L_{2}$-loss. As the value of $d$ becomes large, the convergence rate can become extremely slow, which is a well-known phenomenon called the curse of dimensionality. One possible way to overcome this difficulty is to make additional structural assumptions on the regression function $f\et$ namely to assume that the unknown function $f\et$ is of the form $f_{1}\circ f_{2}$ where $f_{1}$ and $f_{2}$ have some specific structures (e.g., \cite{Stone1985}, \cite{Mammen2007} and \cite{Baraud:2011fk}). For instance, under the generalized additive structure of $f\et$, \cite{Mammen2007} showed that, one can estimate the regression function $f\et$ with rate $n^{-2\alpha/(2\alpha+1)}$ which is independent of the dimension $d$. 

Recently, estimation based on neural networks has demonstrated remarkable success in both experimental and practical domains. Inspiring work has been carried out to systematically analyze the theoretical properties of least squares estimators implemented by various structured neural networks, particularly those employing a ReLU activation function. We mention the work of \cite{Schmidt2020}, \cite{Sophie2021}, \cite{suzuki2019deep} and \cite{Jiao2023}, among others. Based on the established approximation results, these studies have revealed that least squares estimators implemented using appropriate neural network architectures achieve the same minimax convergence rate as that obtained in \cite{Stone1982} when considering a regression function $f\et$ with $\alpha$-smoothness. However, these findings also indicate that without further assumptions on the underlying model, nonparametric regression using deep neural networks is not immune to the curse of dimensionality. Much effort has been devoted to mitigating this issue through network-based estimation approaches (e.g., \cite{SchmidtMani}, \cite{Minshuo2022} and \cite{Nakada2020} where they assume that the distribution of $W$ is supported on a low-dimensional manifold, or the covariates exhibit a low intrinsic dimension, and \cite{Kohler2019}, \cite{suzuki2019iclr}, where structural assumptions are imposed on $f\et$). In particular, it is worth mentioning that, as shown in \cite{Schmidt2020}, neural networks, especially deep ones, exhibit a natural advantage in approximating functions with a compositional structure compared to classical approximation methods.


Given a collection of candidate estimators for $f\et$, most of the aforementioned approaches derive their estimators by minimizing a least-squares-based objective function. While possessing several desirable properties, least squares estimators are highly susceptible to data contamination and the presence of outliers, which are common scenarios encountered in practical applications. To address this issue of instability, several alternative approaches have been proposed in the context of linear regression, such as Huber regression (\cite{Huber1973}), Tukey's biweight regression (\cite{Beaton1974}) and the least absolute deviation regression (\cite{Bassett1978}). 

In the realm of deep learning, a prevailing characteristic is the presence of data abundant in quantity but often deficient in quality. As a result, robustness becomes an essential property to consider when implementing estimation procedures based on deep neural networks (\cite{Barron2019}). However, there has been significantly less research conducted in this field. In \cite{Lederer2020}, upper bounds for the expected excess risks of a specific class of estimators were established. These estimators are obtained by minimizing empirical risk using unbounded, Lipschitz-continuous loss functions on feedforward neural networks, covering cases such as the least absolute deviation loss, Huber loss, Cauchy loss, and Tukey's biweight loss. \cite{Jiao2023} investigated a similar class of estimators. They relaxed several assumptions required in \cite{Lederer2020}, which led to the establishment of their non-asymptotic expected excess risk bounds under milder conditions. They also considered the approximation error introduced by the ReLU neural network and demonstrated that the curse of dimensionality can be mitigated for such class of estimators if the distribution of $W$ is assumed to be supported on an approximately low-dimensional manifold. Drawing upon the approximation results established in \cite{Schmidt2020} and \cite{suzuki2019iclr}, \cite{Padilla2022} examined the properties of quantile regression using deep ReLU neural networks. When the underlying quantile function can be represented as a composition of H\"older functions or when it belongs to a Besov space, they derived convergence rates for the resulting estimators in terms of the mean squared error at the design points. All the previously mentioned work that addresses robust nonparametric regression using deep neural networks assumes the existence of the regression function $f\et$. The approaches they considered and analyzed focus on the robustness under the scenarios where there is a departure from Gaussian distributions to heavy-tailed distributions. When it comes to the case of adversarial attacks, where the statistical model is misspecified from a distributional perspective, their results are unable to provide a theoretical guarantee for the performance of the resulting estimators.

In this paper, we approach the nonparametric regression problem from a novel perspective that acknowledges the possibility of misspecification at the distributional level. We propose a general procedure under mild assumptions to address this problem in a robust manner and investigate its application to ReLU neural networks. Specifically, our primary contributions are as follows.

\begin{enumerate}[label=(\roman*)]
  \item We consider this estimation problem from the perspective of estimating the conditional distributions $Q_{i}\et(W_{i})$ of $Y_{i}$ given $W_{i}$. To handle this statistical issue, we propose an $\ell$-type estimation procedure based on a development of $\ell$-estimation methodology proposed in \cite{Baraud2021}. Our approach is based on the presumption that there exists an underlying function $f\et$ on $\sW$ belonging to some collection $\overline\cF$ such that $Q_{i}\et(W_{i})$ is of the form $Q_{f\et(W_{i})}\sim\cN(f\et(W_{i}),\sigma^{2})$ for all $i\in\{1,\ldots,n\}$. However, our method is not confined to this assumption. In other words, we allow our statistical models to be slightly misspecified: $Q_{i}\et(W_{i})$ may not be exactly of the form $Q_{f\et(W_{i})}$ and even if they were, $f\et$ may not belong to the class $\overline\cF$. 
  \item Assuming that $\overline\cF$ is a VC-subgraph class on $\sW$, we derive a non-asymptotic risk bound for the resulting estimators, measured in terms of the total-variation type distance. Building upon this general result, we offer a comprehensive elucidation of the robustness of our estimators with regard to model misspecification at the distributional level. We also provide a quantitative comparison between the $\ell$-type estimators and another type of robust estimators known as $\rho$-estimators, which were introduced in \cite{Baraud2020}.
  \item We showcase the application of our $\ell$-type estimation procedure using ReLU neural network models. In the case of a well-specified model, we derive uniform risk bounds over H\"older classes for our estimators. By incorporating the lower bounds that we established, we demonstrate that the resulting estimators achieve the minimax optimal rate of convergence.
  
 \item We consider the problem of circumventing the curse of dimensionality by imposing structural assumptions on the underlying regression function $f\et$. More precisely, we assume the function $f\et$ can be expressed as a composition of several H\"older functions, following the consideration in \cite{Schmidt2020}. In contrast to using sparsity-based ReLU neural networks as in \cite{Schmidt2020}, we develop new deep fully-connected ReLU neural networks to approximate composite H\"older functions, enhancing the informativeness of the architectural design. This approximation result can be of independent interest. By leveraging the derived approximation theory, we demonstrate that the $\ell$-type estimators implemented based on appropriate network models can alleviate the curse of dimensionality while converging to the truth at a minimax optimal rate.
 
\end{enumerate}

The paper is organized as follows. In Section~\ref{stat-setting}, we describe our specific statistical framework and set notation. In Section~\ref{ell-construction-section}, we introduce our estimation procedure based on $\ell$-estimation and present our main result regarding the risk bounds for the resulting estimators. We also provide an explanation of why the deviation inequality we establish ensures the desired robustness property of the estimators and compare them with the $\rho$-estimators in that section. In Section~\ref{neural}, we delve into the implementation of our $\ell$-type estimation approach on ReLU neural networks. We establish uniform risk bounds over H\"older classes when the data are truly i.i.d. and the regression function exists. By combining the lower bounds we have derived, we demonstrate the minimax optimality of our estimators under the well-specified scenario. The problem of circumventing the curse of dimensionality is addressed in Section~\ref{curse}, where we impose structural assumptions on the regression function $f\et$. Section~\ref{proof-paper} is devoted to most of the proofs in this paper.

\section{The statistical setting}\label{stat-setting}
Let $X_{i}=(W_{i},Y_{i})$, for $i\in\{1,\ldots,n\}$ be $n$ pairs of independent, but not necessarily i.i.d., random variables with values in a measurable product space $(\sX,\cX)=(\sW\times\sY,\cW\otimes\cY)$. Denote the set of all probabilities on $(\sY,\cY)$ as $\sT$. We assume that the conditional distribution of $Y_{i}$ given $W_{i}=w_{i}$ exists and is given by the value at $w_{i}$ of a measurable function $Q_{i}\et$ from $(\sW,\cW)$ to $\sT$. We endow $\sT$ with the Borel $\sigma$-algebra $\cT$ associated with the total variation distance. Recall that when given two probabilities $P_{1}$ and $P_{2}$ on a measurable space $(A,\sA)$, the total variation distance $\|P_{1}-P_{2}\|_{TV}$ between $P_{1}$ and $P_{2}$ is defined as
\begin{equation*}
\|P_{1}-P_{2}\|_{TV}=\sup_{\cA\in\sA}\cro{P_{1}(\cA)-P_{2}(\cA)}=\frac{1}{2}\int_{A}\left|\frac{dP_{1}}{d\mu}-\frac{dP_{2}}{d\mu}\right|d\mu,
\end{equation*}
where $\mu$ is any reference measure that dominates both $P_{1}$ and $P_{2}$. With this chosen $\cT$, for any $i\in\{1,\ldots,n\}$, the mapping $w\mapsto\|Q_{i}\et(w)-R\|_{TV}$ on $(\sW,\cW)$ is measurable for any probability $R\in\sT$.
%

Given a class of real-valued measurable functions $\overline\cF$ on $\sW$, we presume that, there exists a function $f\et\in\overline\cF$ for which the conditional distributions $Q_{i}\et(W_{i})$ have the structure of $Q_{f\et(W_{i})}=\cN(f\et(W_{i}),\sigma^{2})$ or are at least in close proximity to it. The function $f\et$ is what we refer to as the regression function. It is worth emphasizing, as we mentioned in Section~\ref{intro}, that our statistical model could potentially be misspecified: the conditional distributions $Q_{i}\et(W_{i})$ might not precisely take the form $Q_{f\et(W_{i})}$, or even if they did, the regression function $f\et$ might not belong to the class $\overline\cF$. What we are truly assuming is that the collection $\{Q_{f},\; f\in\overline\cF\}$ provides a suitable approximation of the actual conditional distributions $Q_{i}\et$, for $i\in\{1,\ldots,n\}$. 

Let $\sQ_{\sW}$ represent the collection of all conditional probabilities from $(\sW,\cW)$ to $(\sT,\cT)$, and define $\sbQ_{\sW}=\sQ_{\sW}^{n}$. As a direct result, we obtain the $n$-tuple $\gQ\et=(Q\et_{1},\ldots,Q\et_{n})\in\sbQ_{\sW}$. We equip the space $\sbQ_{\sW}$ with a distance metric resembling the total variation distance.
More precisely, for $\gQ=(Q_{1},\ldots,Q_{n})$ and $\gQ'=(Q'_{1},\ldots,Q'_{n})$ in $\sbQ_{\sW}$, 
\begin{align}
\ell(\gQ,\gQ')&=\frac{1}{n}\E\cro{\sum_{i=1}^{n}\|Q_{i}(W_{i})-Q'_{i}(W_{i})\|_{TV}}\nonumber\\
&=\frac{1}{n}\sum_{i=1}^{n}\int_{\sW}\|Q_{i}(w)-Q'_{i}(w)\|_{TV}dP_{W_{i}}(w).\label{tv-type-loss}
\end{align}
Particularly, when $\ell(\gQ,\gQ')=0$, it signifies that $Q_{i}=Q'_{i}$ $P_{W_{i}}$-a.s., for all $i$.

Building on the $n$ observations $\bsX=(X_{1},\ldots,X_{n})$, we will introduce an estimation approach in the later section to develop an estimator $\widehat f(\bsX)\in\overline\cF$ for the potential regression function $f\et$ (which may not exist). Furthermore, we aim to estimate the $n$-tuple $\gQ\et=(Q\et_{1},\ldots,Q\et_{n})$ by means of the structure $\gQ_{\widehat f}=(Q_{\widehat f},\ldots,Q_{\widehat f})$. We assess the performance of the estimator $\gQ_{\widehat f}$ for $\gQ\et$ through the measure $\ell(\gQ\et,\gQ_{\widehat f})$.

We denote $P = Q\cdot P_{W}$ when $P$ represents the distribution of a random variable $(W,Y)\in\sW\times\sY$, where the marginal distribution of $W$ is $P_{W}$ and the conditional distribution of $Y$ given $W$ is $Q$. One can observe that when $P_{1}=Q_{1}\cdot P_{W}$ and $P_{2}=Q_{2}\cdot P_{W}$, the total variation distance between $P_{1}$ and $P_{2}$ can be represented as
\[
\|P_{1}-P_{2}\|_{TV}=\int_{\sW}\|Q_{1}(w)-Q_{2}(w)\|_{TV}dP_{W}(w).
\]
By defining $P_{i}\et=Q_{i}\et\cdot P_{W_{i}}$ and $P_{i,f}=Q_{f}\cdot P_{W_{i}}$ for a measurable function $f$ that maps $\sW$ to $\R$, we can represent $\ell(\gQ\et,\gQ_{f})$ as the average total variation distance over $n$ samples:
%
\begin{equation}\label{tv-connection-pseudo}
\ell(\gQ\et,\gQ_{f})=\frac{1}{n}\sum_{i=1}^{n}\|P_{i}\et-P_{i,f}\|_{TV}.
\end{equation}


In the case where $W_{i}$ are i.i.d. with the common distribution $P_{W}$ and $Q\et_{i}=Q\et$ for all $i\in\{1,\ldots,n\}$, we may slightly abuse the notation $\ell(Q\et,Q_{\widehat f})$ to measure the distance between $Q\et$ and $Q_{\widehat f}$ defined as 
\begin{equation}\label{n1-ell-def}
\ell(Q\et,Q_{\widehat f})=\int_{\sW}\|Q\et(w)-Q_{\widehat f(w)}\|_{TV}dP_{W}(w).
\end{equation}

%
%

We conclude this section by introducing some notations that will be useful later. We denote $\N^{*}$ the set of all positive natural numbers and $\R_{+}^{*}$ the set of all positive real numbers. For any $x\in\R$, we use the notation $\lfloor x\rfloor$ to represent the largest integer strictly smaller than $x$, and the notation $\lceil x\rceil$ to represent the least integer greater than or equal to $x$. Given any set $J$, we denote its cardinality by $|J|$. For a $\gR\in\sbQ_{\sW}$ and a set $\gA\subset\sbQ_{\sW}$, we define $\ell(\gR,\gA)=\inf_{\gR'\in\gA}\ell(\gR,\gR')$. Unless otherwise specified, $\log$ denotes the logarithm function with base $e$. Let $(E,\cE)$ be a measurable space and $\mu$ be a $\sigma$-finite measure on $(E,\cE)$. For $k\in\cro{1,+\infty}$, we define $\cL_{k}(E,\mu)$ the collection of all the measurable functions $f$ on $(E,\cE,\mu)$ such that $\|f\|_{k,\mu}<+\infty$, where
\begin{equation*}
\|f\|_{k,\mu}=\left\{
\begin{aligned}
&\left(\int_{E}|f|^{k}d\mu\right)^{1/k},&\mbox{for\ }k\in[1,+\infty),\\
&\inf\{K>0,\;|f|\leq K\;\mu-\mbox{a.e.}\},&\mbox{for\ }k=\infty.
\end{aligned}
\right.
\end{equation*}
We denote the associated equivalent classes as $\L_{k}(E,\mu)$ where any two functions coincide for $\mu$-a.e. can not be distinguished. In particular, we write the norm $\|\cdot\|_{k}$ with $k\in\cro{1,+\infty}$ when $\mu=\lambda$ is the Lebesgue measure. Throughout the paper, $c$ or $C$ denotes positive numerical constant which may vary from line to line.

\section{$\ell$-Type estimation under regression setting}\label{ell-construction-section}
We employ an $\ell$-type estimator, drawing inspiration from the concepts outlined in a series of papers presented in \cite{Baraud2021} within a general framework, as well as from the content of \cite{Baraud2022}, which is specifically dedicated to density estimation.
Consider a set of $n$ independent random variables denoted as $X_{1},\ldots,X_{n}$, where their values are drawn from a measured space $(\sX,\cX)$. In essence, $\ell$-estimation offers a versatile approach to acquiring a robust estimator for the actual joint distribution $\gP\et$ of $X_{1},\ldots,X_{n}$. The established $\ell$-estimation approach begins by introducing a set of potential probabilities $\overline\sbP$, intended to offer a suitable approximation of $\gP\et$. The primary challenge in implementing $\ell$-estimation within a regression framework lies in the absence of information concerning the marginal distributions $P_{W_{i}}$ required for constructing candidate probabilities and designing the estimation procedure. Moreover, our objective does not encompass the task of estimating these marginal distributions. In this scenario, further effort is necessary to implement $\ell$-type estimation and establish a risk bound for the resulting estimator.

\subsection{Constructing the $\ell$-type estimator}
Let $\overline\cF$ be a collection of real-valued measurable functions on $\sW$, which we call it a model. For any $f\in\overline\cF$, we denote $Q_{f}$ the conditional Gaussian distribution induced by the function $f$, i.e., given any $w\in\sW$, $Q_{f(w)}$ is a normal distribution centered around $f(w)$, with a variance of $\sigma^2$, and denote $q_{f(w)}$ the density function of the Gaussian distribution $Q_{f(w)}$ with respect to the Lebesgue measure. To prevent any measurability issue, we introduce the notation $\cF$, representing either a finite or, at most, a countable subset of $\overline\cF$. Subsequently, the majority of our discussion will be focused on the set $\cF$. Nevertheless, as we delve into further details, it turns out that through careful choice of $\cF$, no approximation power will be sacrificed in comparison to estimations based on $\overline{\cF}$.

Given $f_{1},f_{2}\in\cF$, we define for any $(w,y)\in\sW\times\sY$, $$t_{(f_{1},f_{2})}(w,y)=
\1_{q_{f_{2}(w)}(y)>q_{f_{1}(w)}(y)}-Q_{f_{1}(w)}\left(q_{f_{2}(w)}>q_{f_{1}(w)}\right).$$
Employing the function $t_{(f_{1},f_{2})}(\cdot,\cdot)$ produces the following inequalities.
\begin{lem}\label{t-property}
Let $P\et=Q\et\cdot P_{W}$ represent the distribution of a pair of random variables $(W,Y)\in\sW\times\sY$, where the first marginal distribution is $P_{W}$, and the conditional distribution of $Y$ given $W$ is denoted by $Q\et$. For any $f_{1},f_{2}\in\cF$, any $P_{W}$ and any $Q\et\in\sQ_{\sW}$, we have
\begin{equation}\label{t-preproperty}
\ell(Q_{f_{1}},Q_{f_{2}})-\ell(Q\et,Q_{f_{2}})\leq\E_{P\et}\cro{t_{(f_{1},f_{2})}(W,Y)}\leq\ell(Q\et,Q_{f_{1}}).
\end{equation}
\end{lem}
The proof of Lemma~\ref{t-property} is deferred to Section~\ref{Lem-t-property}. Lemma~\ref{t-property} implies that the family of test statistics $t_{(f_{1},f_{2})}$ holds information concerning the $\ell$-type distance between two of $Q_{f_{1}}$, $Q_{f_{2}}$, and $Q\et$, which is an essential property for constructing our final estimator.

For any $f_{1},f_{2}\in\cF$ and $n$ pairs of observations $\bsX=(X_{1},\ldots,X_{n})$ with $X_{i}=(W_{i},Y_{i})$, $i\in\{1,\ldots,n\}$, we design the function
\[
\gT_{l}(\bsX,f_{1},f_{2})=\sum_{i=1}^{n}t_{(f_{1},f_{2})}(W_{i},Y_{i})
\]
and set
\[
\gT_{l}(\bsX,f_{1})=\sup_{f_{2}\in\cF}\gT_{l}(\bsX,f_{1},f_{2}).
\]
Our final estimator of $\gQ\et=(Q\et_{1},\ldots,Q\et_{n})$ is defined as $\gQ_{\widehat f}=(Q_{\widehat f},\ldots, Q_{\widehat f})$, where $\widehat f(\bsX)$ is an $\epsilon$-minimizer over $\cF$ of the map
$f_{1}\mapsto\gT_{l}(\bsX,f_{1})$. More precisely, given $\epsilon>0$, the $\ell$-type estimator within the set $\cF$ is defined as any measurable function $\widehat f(\bsX)$ of the random (and non-void) set 
\begin{equation*}\label{def-l}
\sE(\bsX,\epsilon)=\left\{f\in\cF,\;\gT_{l}(\bsX,f)\leq\inf_{f'\in\cF}\gT_{l}(\bsX,f')+\epsilon\right\}.
\end{equation*}
\begin{rmk}
The parameter $\epsilon$ is devised to ensure the existence of the estimator $\widehat f$. As we will explore in Section~\ref{performance-ell}, it is prudent to choose a relatively small value for $\epsilon$, specifically not significantly greater than 1, as this choice improves the risk bound of an $\ell$-type estimator. Specifically, when a function $f\in\cF$ exists such that $\gT_{l}(\bsX,f)=\inf_{f'\in\cF}\gT_{l}(\bsX,f')$, it is advisable to prioritize this $f$ as the estimator $\widehat f$.

Furthermore, considering that $\gT_{l}(\bsX,f)\geq\gT_{l}(\bsX,f,f)=0$ for all $f\in\cF$, any function $\widehat f\in\cF$ meeting the condition $0\leq\gT_{l}(\bsX,\widehat f)\leq\epsilon$ qualifies as an $\ell$-type estimator.

\end{rmk}
\subsection{The performance of the $\ell$-type estimator}\label{performance-ell}
Before delving into the theoretical performance of our $\ell$-type estimator, we lay the foundation by stating our main assumption on the model $\overline\cF$. To facilitate this, we introduce the following definition:
\begin{df}[VC-subgraph]\label{vcsubgraph}
An (open) subgraph of a function $f$ in $\overline\cF$ is the subset of $\sW\times\R$ given by
\begin{equation*}
\sC_{f}=\left\{(w,u)\in\sW\times\R,\;f(w)>u\right\}.
\end{equation*}
A collection $\overline\cF$ of real-valued measurable functions on $\sW$ is VC-subgraph with dimension not larger than $V$ if, for any finite subset $\cS\subset\sW\times\R$ with $|\cS|=V+1$, there exists at least one subset $S$ of $\cS$ such that for any $f\in\overline\cF$, $S$ is not the intersection of $\cS$ with $\sC_{f}$, i.e.
\begin{equation*}
S\ne \cS\cap\sC_{f}\quad \text{whatever $f\in\overline\cF$.}
\end{equation*}
\end{df}
Herein, we proceed to introduce our primary assumption concerning the model $\overline\cF$.
\begin{ass}\label{vcass}
The class of functions $\overline\cF$ is VC-subgraph on $\sW$ with dimension not larger than $V\geq1$.
\end{ass}
Encompassing a range of widely employed examples, Assumption~\ref{vcass} is formulated under a considerably broad scope. For instance, when $\overline\cF$ is contained in a linear space with finite dimension $D$, Assumption~\ref{vcass} is fulfilled with $V=D+1$ according to Lemma~2.6.15 of \cite{MR1385671}. Moreover, when $\overline{\cF}$ represents a fully connected ReLU neural network, it has been demonstrated in \cite{BarlettJMLR} [Theorem~7] that the VC-dimension of $\overline{\cF}$ is linked to the depth and width of the network. Further elaboration on $\ell$-estimation based on neural networks will be provided in Section~\ref{neural} and Section~\ref{curse}.

Building upon Assumption~\ref{vcass}, we can establish the following non-asymptotic exponential inequalities for the upper deviations of a total variation type distance between the true distribution of the data and the estimated one based on $\widehat f(\bsX)$.
\begin{thm}\label{general}
Under Assumption~\ref{vcass}, whatever the conditional distributions $\gQ\et=(Q\et_{1},\ldots,Q\et_{n})$ of the $Y_{i}$ given $W_{i}$ and the distributions of $W_{i}$, any $\ell$-type estimator $\widehat f$ based on the class $\cF$ satisfies that for any $\overline f\in\cF$ and any $\xi>0$, with a probability at least $1-e^{-\xi}$,
\begin{equation}\label{exp-inequa-everypoint}
\ell(\gQ_{\overline f},\gQ_{\widehat f})\leq2\ell(\gQ\et,\gQ_{\overline f})+291.2\sqrt{\frac{V}{n}}+14573.4\frac{V}{n}+\sqrt{\frac{8(\xi+\log2)}{n}}+\frac{\epsilon}{n}.
\end{equation}
In particular, with the triangle inequality,
\begin{equation}\label{exp-inequa}
\ell(\gQ\et,\gQ_{\widehat f})\leq3\ell(\gQ\et,\sbQ)+291.2\sqrt{\frac{V}{n}}+14573.4\frac{V}{n}+\sqrt{\frac{8(\xi+\log2)}{n}}+\frac{\epsilon}{n},
\end{equation}
where $\sbQ=\{\gQ_{f},\; f\in\cF\}$. As a consequence of \eref{exp-inequa}, for any $n\geq V$, integration with respect to $\xi>0$ yields the following risk bound for the resulting estimator $\gQ_{\widehat f}=(Q_{\widehat f},\ldots,Q_{\widehat f})$ 
\begin{equation}\label{riskbound}
\E\cro{\ell(\gQ\et,\gQ_{\widehat f})}\leq C_{\epsilon}\cro{\ell(\gQ\et,\sbQ)+\sqrt{\frac{V}{n}}},
\end{equation}
where $C_{\epsilon}>0$ is a numerical constant depending on $\epsilon$ only.
\end{thm}
The proof of Theorem~\ref{general} is deferred to Section~\ref{general-thm-proof}. Let us now provide some remarks regarding this result.
\begin{rmk}
Consider the set $\overline\sbQ=\{\gQ_{f},\;f\in\overline\cF\}$. It is clear that if $\sbQ$ is dense in $\overline\sbQ$ with respect to the (pseudo) distance $\ell$, both \eref{exp-inequa} and \eref{riskbound} also remain valid when replacing $\sbQ$ with $\overline\sbQ$. This is the situation in which the subset $\cF$ is dense in $\overline\cF$ with respect to the topology of pointwise convergence. For further insights in this direction, we refer to Section 4.2 of \cite{Baraud2018}. For the sake of simplicity in our explanation, let us temporarily assume in this section that $\sbQ$ is dense in $\overline\sbQ$ with respect to $\ell$.

\end{rmk}

\begin{rmk}
According to \eref{riskbound}, the risk of the resulting estimator is bounded, up to a numerical constant, by the sum of two terms. The term $\ell(\gQ\et,\overline\sbQ)$ corresponds to the approximation error incurred by employing the model $\overline\cF$, while $\sqrt{V/n}$ illustrates the complexity of the considered model $\overline\cF$. Hence, a suitable model $\overline\cF$ should strike a balance between these two factors, namely, a model that is not excessively complex yet offers a good approximation of the underlying regression function.
\end{rmk}
\begin{rmk}\label{ideal-situation}
In the favourable situation where the data $X_{i}=(W_{i},Y_{i})$ are truly i.i.d. with $Q\et_{i}=Q_{f\et}$, $i\in\{1,\ldots,n\}$ for some $f\et\in\overline\cF$, we can deduce from \eref{riskbound} that
\begin{equation*}
\E\cro{\ell(Q_{f\et},Q_{\widehat f})}\leq C_{\epsilon}\sqrt{\frac{V}{n}}.
\end{equation*}
In typical situations, the value of $V$ aligns with the magnitude of parameters necessary to parametrize $\overline\cF$, which cannot be improved in general. As we shall observe in Section~\ref{holder-risk}, the above risk bound will lead to an optimal rate of convergence in the minimax sense when the regression function $f\et$ is assumed to be a smooth function of regularity $\alpha$.
\end{rmk}

\begin{rmk}
The term $\ell(\gQ\et,\overline\sbQ)$ elucidates the robustness property of the resulting estimator concerning model misspecification. To illustrate, let us consider the general scenario where the data are only independent and the true joint distribution is given by
\begin{equation}\label{formulation}
\gP\et=\bigotimes_{i=1}^{n}P\et_{i}=\bigotimes_{i=1}^{n}\left[(1-\beta_{i})P_{i,\overline f}+\beta_{i} R_{i}\right],\quad\quad\sum_{i=1}^{n}\beta_{i}\leq\frac{n}{2},
\end{equation}
with some $\overline f\in\overline \cF$, $P_{i,\overline f}=Q_{\overline f}\cdot P_{W_{i}}$, $R_{i}$ being an arbitrary distribution on $\sX=\sW\times\sY$ and $\beta_{i}$ taking values in $\cro{0,1}$ for all $i\in\{1,\ldots,n\}$. 
With the connection \eref{tv-connection-pseudo} between the pseudo distance $\ell$ and $\|\cdot\|_{TV}$, we can deduce from \eref{riskbound} that 
\begin{align}
\E\cro{\frac{1}{n}\sum_{i=1}^{n}\|P\et_{i}-P_{i,\widehat f}\|_{TV}}&\leq C_{\epsilon}\cro{\frac{1}{n}\sum_{i=1}^{n}\|P\et_{i}-P_{i,\overline f}\|_{TV}+\sqrt{\frac{V}{n}}}\nonumber\\
&\leq C_{\epsilon}\cro{\frac{1}{n}\sum_{i=1}^{n}\beta_{i}+\sqrt{\frac{V}{n}}},\label{misspecification-bound}
\end{align}
where the second inequality comes from the fact that $\|\cdot\|_{TV}$ is bounded by 1. The above result implies that as long as the quantity $(\sum_{i=1}^{n}\beta_{i})/n$ remains small compared to the term $\sqrt{V/n}$, the performance of the resulting estimator will not deteriorate significantly in comparison to the ideal situation presented in Remark~\ref{ideal-situation}.

The formulation \eref{formulation} can be utilized to provide a more detailed explanation of the stability of $\ell$-type estimation procedure. More precisely, in the case of the presence of outliers, the observations include several outliers, the indices of which are marked as a non-empty subset $J$ of $\{1,\ldots,n\}$. For any $i\in J$, $R_{i}=\delta_{a_{i}}$ and $\beta_{i}=\1_{i\in J}$ for all $i\in\{1,\ldots,n\}$. The bound \eref{misspecification-bound} indicates that our estimation procedure remains stable as long as $|J|/n$ remains small compared to $\sqrt{V/n}$. This accounts for the robustness when the outliers present. Under another scenario, where the data are contaminated, $(W_{i}, Y_{i})$ are i.i.d., and $P_{W_{i}}=P_{W}$. A portion $\beta\in(0,1/2]$ of the $n$ samples is drawn according to an arbitrary distribution $R_{i}=R$ (where $R$ is not equal to $P_{i,\overline f}$), while the remaining part follows the distribution $P_{i,\overline f}$. In this case, as an immediate consequence of \eref{misspecification-bound}, the performance of our estimator remains stable as long as the contamination proportion $\beta$ remains small compared to the value of $\sqrt{V/n}$.
\end{rmk}

\subsection{Connection to $\L_{1}$-distance between the regression functions}\label{l1-risk-ell-connect}
As we have seen in Section~\ref{performance-ell}, we establish non-asymptotic inequalities for the upper deviations of a total variation type distance between the true conditional distributions and the estimated one based on $\widehat f$. In the context of a regression setting where the data are truly i.i.d. and follow the common marginal distribution $P_{W}$, the function $f\et$ exists, such that $Q_{i}\et=Q_{f\et}$. It would be interesting to investigate the performance of the $\ell$-type estimator $\widehat{f}(\bsX)$ in relation to the regression function $f\et$, utilizing a suitable distance metric, as typically considered in the literature. Given two real-valued functions $f$ and $f'$ on $\sW$, it turns out that $\ell(Q_{f},Q_{f'})$ can be related to the $\L_{1}(P_{W})$-distance between $f$ and $f'$. We present the result as follows.
\begin{lem}\label{l1-ell}
For any two measurable real-valued functions $f,f'$ on $\sW$, and any $w\in\sW$, we have
\begin{equation}\label{tv-equality}
\|Q_{f(w)}-Q_{f'(w)}\|_{TV}=1-2\Phi\left(-\frac{|f'(w)-f(w)|}{2\sigma}\right),
\end{equation}
where the notation $\Phi$ stands for the cumulative distribution function of the standard normal distribution. Consequently,
\begin{equation}\label{tvconnect}
0.78\min\left\{\frac{\|f-f'\|_{1,P_{W}}}{\sqrt{2\pi}\sigma},1\right\}\leq\ell(Q_{f},Q_{f'})\leq\min\left\{\frac{\|f-f'\|_{1,P_{W}}}{\sqrt{2\pi}\sigma},1\right\}.
\end{equation}
\end{lem}
\begin{proof}
For any two probabilities $P$ and $R$ on the measured space $(\sX,\cX)$, it is well known that the total variation distance can equivalently be written as $\|P-R\|_{TV}=R(r>p)-P(r>p),$ where $p$ and $r$ stand for the respective densities of $P$ and $R$ with respect to some common dominating measure $\mu$. Therefore, a fundamental calculation reveals that for any $w\in\sW$, 
\begin{align}
\|Q_{f(w)}-Q_{f'(w)}\|_{TV}&=Q_{f'(w)}\left(q_{f'(w)}>q_{f(w)}\right)-Q_{f(w)}\left(q_{f'(w)}>q_{f(w)}\right)\nonumber\\
&=\cro{1-\Phi\left(-\frac{|f'(w)-f(w)|}{2\sigma}\right)}-\Phi\left(-\frac{|f'(w)-f(w)|}{2\sigma}\right)\nonumber\\
&=1-2\Phi\left(-\frac{|f'(w)-f(w)|}{2\sigma}\right),\label{gaussian-accu}
\end{align}
which concludes the equality \eref{tv-equality}. We also note from \eref{gaussian-accu} that
\[
\|Q_{f(w)}-Q_{f'(w)}\|_{TV}=\P\cro{|Z|\leq\frac{1}{2}\left(\frac{|f'(w)-f(w)|}{\sigma}\right)}, 
\]
where $Z$ is a standard real-valued Gaussian random variable. Recall that 
\begin{equation}\label{def-ell-int}
\ell(Q_{f},Q_{f'})=\int_{\sW}\|Q_{f(w)}-Q_{f'(w)}\|_{TV}dP_{W}(w).
\end{equation}
Based on \eref{def-ell-int}, the conclusion of \eref{tvconnect} follows by applying Lemma~1 in \cite{Baraud2021} with $d=1$ and replacing $|m-m'|$ with $(|f'(w)-f(w)|)/\sigma$.
\end{proof}
The above result indicates that when the two functions $f$ and $f'$ are sufficiently close to each other with respect to the $\L_{1}(P_{W})$-distance, the quantity $\ell(Q_{f},Q_{f'})$ is of order $\|f-f'\|_{1,P_{W}}/(\sqrt{2\pi}\sigma)$. Conversely, when $f$ and $f'$ are far apart, the value of $\ell(Q_{f},Q_{f'})$ remains approximately of the order of 1. Combining Lemma~\ref{l1-ell} with \eref{riskbound}, we can deduce that 
\begin{align}
&\min\left\{\frac{\E\cro{\|f\et-\widehat f\|_{1,P_{W}}}}{\sqrt{2\pi}\sigma},1\right\}
\leq C_{\epsilon}\cro{\inf_{f\in\cF}\ell(Q_{f\et},Q_{f})+\sqrt{\frac{V}{n}}},\label{l1-riskbound-min}
\end{align}
where $C_{\epsilon}>0$ is a numerical constant depending on $\epsilon$ only. As we shall see it later, in typical applications, if we can find a nice model to approximate the regression function $f\et$ in the sense that the right hand of \eref{l1-riskbound-min} is smaller than 1, then we finally obtain a risk bound for $\widehat f(\bsX)$ with respect to the $\L_{1}$-distance:
\begin{equation*}
\E\cro{\|f\et-\widehat f\|_{1,P_{W}}}\leq C_{\epsilon,\sigma}\cro{\inf_{f\in\cF}\|f\et-f\|_{1,P_{W}}+\sqrt{\frac{V}{n}}},
\end{equation*}
where $C_{\epsilon,\sigma}$ is a numerical constant depending on $\epsilon,\sigma$ only. 
\subsection{Comparison with $\rho$-estimation}
As mentioned in Section~\ref{performance-ell}, one notable feature of the $\ell$-type estimators is their robustness under misspecification. Interestingly, the $\rho$-estimators also exhibit robustness properties, but they are quantified using a Hellinger-type distance, rather than the one based on the total variation distance. For a more comprehensive understanding of the $\rho$-estimation methodology, one can refer to \cite{Baraud2018} and \cite{Baraud2020}, with the latter primarily focusing on the regression setting. It is worth noting that while there exists some connection between the Hellinger distance and the total variation distance, they are not equivalent in general. The main distinctions between these two types of estimators has been examined in Section~7.1 of \cite{Baraud2021} which includes an illustration of regression under a fixed design. Leveraging the results we have established in Section~\ref{performance-ell} and \ref{l1-risk-ell-connect}, we are therefore able to delve deeper in this direction, especially under a random regression design setting.

To illustrate simply, we assume the data are truly i.i.d. with $P_{W_{i}}=P_{W}$ and $Q\et_{i}=Q\et$, for all $i\in\{1,\ldots,n\}$. We write $P\et=Q\et\cdot P_{W}$ the true distribution of $(W,Y)\in\sW\times\sY$. For some $f\in\cF$, provided the term $\ell(Q\et,Q_{f})$ and $1/n$ are both sufficiently small, employing Lemma~\ref{l1-ell}, one can deduce from \eref{exp-inequa-everypoint} that the $\ell$-type estimator $\widehat f_{\ell}(\bsX)$ satisfies
\begin{equation}\label{tv-result}
\E\cro{\|f-\widehat f_{\ell}\|_{1,P_{W}}}\leq C_{\sigma,\epsilon}\cro{\|P\et-P_{f}\|_{TV}+\sqrt{\frac{V}{n}}},
\end{equation}
where $P_{f}=Q_{f}\cdot P_{W}$. From another point of view, we can deduce, through a slight modification of Theorem~1 in \cite{Baraud2020}, that for any $f\in\cF$, the $\rho$-estimator $\widehat f_{\rho}(\bsX)$ complies with the following
\begin{equation*}
\E\cro{h^{2}(P_{f},P_{\widehat f_{\rho}})}\leq C\cro{h^{2}(P\et,P_{f})+\frac{V(1+\log n)}{n}},
\end{equation*}
where $C>0$ is some numerical constant and $h$ stands for the Hellinger distance. Considering the fact that 
\begin{align*}
h^{2}(P_{f},P_{\widehat f_{\rho}})&=\int_{\sW}1-\exp\cro{-\frac{|f(w)-\widehat f_{\rho}(w)|^{2}}{8\sigma^{2}}}dP_{W}(w)\\
&\geq(1-e^{-1})\left(\frac{\|f-\widehat f_{\rho}\|^{2}_{2,P_{W}}}{8\sigma^{2}}\wedge1\right),
\end{align*}
we can deduce, using H\"older's inequality and a similar argument to that used in obtaining \eref{tv-result}, that for some $f\in\cF$, given the term $h^{2}(P\et,P_{f})$ and $1/n$ are both sufficiently small
\begin{equation}\label{rho-result}
\E\cro{\|f-\widehat f_{\ell}\|_{1,P_{W}}}\leq C_{\sigma}\cro{h(P\et,P_{f})+\sqrt{\frac{V(1+\log n)}{n}}}.
\end{equation}

If we put the numerical constants $C_{\sigma,\epsilon}$, $C_{\sigma}$ aside, the main difference between the two risk bounds lie in the fact that they express the robustness of the two different estimators by the approximation term $\|P\et-P_{f}\|_{TV}$ and $h(P\et,P_{f})$ respectively. With the connection that for any two probabilities $P_{1}, P_{2}$,
$$\|P_{1}-P_{2}\|_{TV}\leq\sqrt{2}h(P_{1},P_{2}),$$ we can conclude that the stability of the $\ell$-type estimators will not be significantly worse than the $\rho$-estimators. In fact, the $\ell$-type estimators can posses much more robustness than the $\rho$-estimators. To explain it in details, consider the misspecified formulation:
$$P\et=(1-\beta)P_{f\et}+\beta R,\quad\mbox{for some small\ }\beta\in(0,1),$$ where $f\et\in\cF$, $R\not=P_{f\et}$ is any arbitrary distribution on $\sX=\sW\times\sY$. On the one hand, we can calculate that
\begin{equation}\label{stable-ell-order}
\|P\et-P_{f\et}\|_{TV}=\beta\|P_{f\et}-R\|_{TV},
\end{equation}
which is of the order of magnitude $\beta$. On the other hand, we have
\begin{equation}\label{stable-rho-order}
h(P\et,P_{f\et})\leq\sqrt{1-\sqrt{1-\beta}},
\end{equation}
which is at most of the order of magnitude $\sqrt{\beta/2}$. Therefore, for small values of $\beta$, the above computation indicates that the term $h(P\et,P_{f\et})$ is much larger than $\|P\et-P_{f\et}\|_{TV}$. Combining \eref{tv-result} with \eref{stable-ell-order}, we deduce that the $\ell$-type estimators remain stable as long as $\beta$ is small as compared to $1/\sqrt{n}$. Combining \eref{rho-result} with \eref{stable-rho-order}, we know that the performance of the $\rho$-estimators deteriorate immediately as long as $\beta$ becomes large as compared to $(\log n)/n$. This analysis implies that the $\ell$-type estimators possess more robustness compared to the ones obtained from $\rho$-estimation.

\section{Applications of $\ell$-type Estimation Using Neural Networks}\label{neural}
In recent years, experimental findings have demonstrated the significant success of neural networks modeling in various applications. From a theoretical perspective, it has been observed that neural networks, especially deep ones (see, for example, \cite{Schmidt2020} and \cite{suzuki2019deep}), possess a natural advantage over classical methods when approximating functions with specific characteristics. In this section, we will discuss $\ell$-type estimation for models based on neural networks. The covariates $W_{i}$ are assumed to be i.i.d. on $\sW=\cro{0,1}^{d}$, following the common distribution $P_{W}$, while $Q_{i}\et=Q\et$ holds for all $i\in\{1,\ldots,n\}$. 
\subsection{ReLU feedforward neural networks}
We start with introducing some preliminaries of the ReLU feedforward neural networks. Recall the Rectifier Linear Unit (ReLU) activation function $\sigma:\R\rightarrow\R$, which is defined as $$\sigma(x)=\max\{0,x\}.$$ For any vector ${\bs{x}}=(x_{1},\ldots,x_{p})^{\top}\in\R^{p}$, where $p\in\N^{*}$, the notation $\sigma({\bs{x}})$ represents the activation function applied component-wise, defined as follows: $$\sigma({\bs{x}})=(\max\{0,x_{1}\},\ldots,\max\{0,x_{p}\})^{\top}.$$ 

A fundamental and extensively employed type of feedforward neural networks in practice is the multi-layer perceptrons, where the neurons in consecutive layers are fully connected through linear transformation matrices. In our later discussion on applying the $\ell$-type estimation, we will focus on the multi-layer perceptrons with ReLU activation function. To begin, let's introduce the expression of the multi-layer perceptrons under consideration. For any vector ${\bs{p}}=(p_{0},\ldots,p_{L+1})\in(\N^{*})^{L+2}$ with $p_{0}=d$ and $p_{L+1}=1$ and $L\in\N^{*}$, we denote the multi-layer perceptron $\overline\cF_{(L, {\bs{p}})}$ as a collection of functions of the form:
\begin{equation*}
f:\R^{d}\rightarrow\R,\quad {\bs{w}}\mapsto f({\bs{w}})=M_{L}\circ\sigma\circ M_{L-1}\circ\cdots\circ\sigma\circ M_{0}({\bs{w}}),
\end{equation*}
where $$M_{l}({\bs{y}})=A_{l}({\bs{y}})+b_{l},\mbox{\quad for\ }l=0,\ldots,L,$$ $A_{l}$ is a $p_{l+1}\times p_{l}$ weight matrix and the shift vector $b_{l}$ is of size $p_{l+1}$ for any $l\in\{0,\ldots,L\}$. In the first layer, the input data consists of the values of the predictor $W$, whereas the last layer represents the output. With the expression given above, we say that the network $\overline\cF_{(L, {\bs{p}})}$ comprises $L$ hidden layers and a total of $(L+2)$ layers. For $l\in\{1,\ldots,L\}$, we refer to $p_{l}$ as the width of the $l$-th hidden layer. The entries in these weight matrices and vectors typically vary in $\R$ or a subinterval of $\R$, which is what we refer to as parameters. In the latter scenario, we employ the notation $\overline\cF_{(L,{\bs{p}},K)}\subset\overline\cF_{(L,{\bs{p}})}$, denoting the set of all functions with parameters ranging within the interval $\cro{-K,K}$.
Furthermore, we use the notation $\cF_{(L, {\bs{p}})}$ (or $\cF_{(L, {\bs{p}},K)}$) for the multi-layer perceptron, which shares the same architecture as $\overline\cF_{(L, {\bs{p}})}$ (or $\overline\cF_{(L, {\bs{p}},K)}$ respectively), but with the distinction that all the parameters take values in $\Q$. In some of our application scenarios, it suffices to consider a multi-layer perceptron with a rectangular design, where $p_{l}=p$ for all $l\in\{1,\ldots,L\}$. In this case, we may use the simplified notation $\cF_{(L,p)}$ (or $\cF_{(L,p,K)}$) to represent the class $\cF_{(L,{\bs{p}})}$ (or $\cF_{(L,{\bs{p}},K)}$ respectively) for ${\bs{p}}=(d,p,\ldots,p,1)$. 

We discuss the implementation of the $\ell$-type estimation on ReLU neural networks $\cF_{(L, {\bs{p}},K)}$. To implement the procedure introduced in Section~\ref{ell-construction-section}, we work on the countable subset $\cF_{(L, {\bs{p}},K)}$ of the model $\overline\cF_{(L, {\bs{p}},K)}$. We can establish the following result.

\begin{lem}\label{network-dense}
For any $L\in\N^{*}$, ${\bs{p}}=(p_{0},\ldots,p_{L+1})\in(\N^{*})^{L+2}$ with $p_{0}=d$, $p_{L+1}=1$ and a finite positive constant $K$, the class of functions $\cF_{(L,{\bs{p}},K)}$ is dense in $\overline\cF_{(L,{\bs{p}},K)}$ with respect to the supremum norm $\|\cdot\|_{\infty}$.
\end{lem}
The proof of Lemma~\ref{network-dense} is postponed to Section~\ref{network-dense-supnorm}. Lemma~\ref{network-dense} ensures that our estimation approach applied to the countable model $\cF_{(L,{\bs{p}},K)}$ does not compromise approximation power compared to $\overline\cF_{(L,{\bs{p}},K)}$.

The following proposition establishes VC-dimensional bounds for rectangular multi-layer perceptrons employing a ReLU activation function. This result can be derived from Proposition 5 in \cite{ChenModSel}, which also aligns with those stated in Theorem 7 of \cite{BarlettJMLR}.

\begin{prop}\label{vc-dim}
For any $L\in\N^{*}$, $p\in\N^{*}$, the class of functions $\overline\cF_{(L,p)}$ is a VC-subgraph on $\sW$ with dimension 
\begin{equation}\label{vc-infty-bound}
V(\overline\cF_{(L,p)})\leq(L+1)\left(s+1\right)\log_{2}\cro{2\left(2e(L+1)\left(\frac{pL}{2}+1\right)\right)^{2}},
\end{equation}
where $s=p^{2}(L-1)+p(L+d+1)+1$.
\end{prop}
This result shows the connection between the VC-dimensional bounds and the depth and width of ReLU rectangular multi-layer perceptrons. Specifically, for any finite constant $K>0$, as $\overline\cF_{(L,p,K)}\subset\overline\cF_{(L,p)}$, the dimensional bounds \eref{vc-infty-bound} also apply to the class $\overline\cF_{(L,p,K)}$. We will use Proposition~\ref{vc-dim} along with other results to derive the risk bounds for the $\ell$-type estimators when applying our approach to ReLU feedforward neural networks.\subsection{Approximating functions in H\"older space}\label{holder-risk}
In this section, we examine the performance of the $\ell$-type estimators implemented on the ReLU feedforward neural networks. We consider the regression setting, where the regression function $f\et$ exists, and we assume that it belongs to an $\alpha$-smoothness H\"older class.

Given $t\in\N^{*}$ and $\alpha\in\R_{+}^{*}$, we define $\cH^{\alpha}(D,B)$ an $\alpha$-H\"older ball with radius $B$ as the collection of functions $f:D\subset\R^{t}\rightarrow\R$ such that $$\max_{\substack{{\bs{\beta}}=(\beta_{1},\ldots,\beta_{t})^{\top}\in\N^{t}\\ \sum_{j=1}^{t}\beta_{j}\leq\lfloor\alpha\rfloor}}\|\partial^{\bs{\beta}}f\|_{\infty}\leq B\quad\mbox{and}\ \max_{\substack{{\bs{\beta}}\in\N^{t}\\ \sum_{j=1}^{t}\beta_{j}=\lfloor\alpha\rfloor}}\sup_{\substack{{\bs{x}},{\bs{y}}\in D\\{\bs{x}}\not={\bs{y}}}}\frac{\left|\partial^{\bs{\beta}}f({\bs{x}})-\partial^{\bs{\beta}}f({\bs{y}})\right|}{\|{\bs{x}}-{\bs{y}}\|_{2}^{\alpha-\lfloor\alpha\rfloor}}\leq B,$$ where for any ${\bs{\beta}}=(\beta_{1},\ldots,\beta_{t})^{\top}\in\N^{t}$, $\partial^{\bs{\beta}}=\partial^{\beta_{1}}\cdots\partial^{\beta_{t}}$.

Based on the notation introduced above, in this section, we assume that $Q\et=Q_{f\et}$, where $f\et\in\cH^{\alpha}(\cro{0,1}^{d},B)$, with a specified smoothness index $\alpha\in\R_{+}^{*}$ and a finite constant $B>0$. For any $\alpha\in\R_{+}^{*}$, the following result demonstrates the error introduced by various ReLU neural networks when approximating the class $\cH^{\alpha}(\cro{0,1}^{d},B)$, as deduced from Corollary~3.1 of \cite{Jiao2023}.

\begin{prop}\label{appro-holder}
Assume that $f\in\cH^{\alpha}(\cro{0,1}^{d},B)$ with $\alpha\in\R_{+}^{*}$ and a finite constant $B>0$. For any $M,N\in\N^{*}$, there exists a function $\overline f$ implemented by a ReLU neural network $\overline\cF_{(L,p)}$ with a width of $$p=38(\lfloor\alpha\rfloor+1)^{2}3^{d}d^{\lfloor\alpha\rfloor+1}N\lceil\log_{2}(8N)\rceil$$ and a depth of $$L=21(\lfloor\alpha\rfloor+1)^{2}M\lceil\log_{2}(8M)\rceil+2d$$ such that
\[
\big|f({\bs{w}})-\overline f({\bs{w}})\big|\leq19B(\lfloor\alpha\rfloor+1)^{2}d^{\lfloor\alpha\rfloor+\frac{\alpha\vee1}{2}}(NM)^{-\frac{2\alpha}{d}},
\]
for all ${\bs{w}}\in\cro{0,1}^{d}$.
\end{prop}
In fact, several approximation results have been established regarding the H\"older class of smoothness functions, for instance in \cite{Minshuo2019}, \cite{Schmidt2020}, and \cite{Nakada2020}, among others. The reason why we consider using Proposition~\ref{appro-holder} is mainly due to two aspects. Firstly, unlike most existing results where the prefactor in the error bound depends exponentially on the dimension $d$, the prefactor in this error bound depends only polynomially on the dimension $d$. Secondly, it offers specific structures of the neural networks to be considered, thus making the result more informative.

Building upon the results of Theorem~\ref{general}, Proposition~\ref{vc-dim}, \ref{appro-holder}, and Lemma~\ref{network-dense}, we derive the risk bounds for the $\ell$-type estimators implemented by different networks as follows.
\begin{cor}\label{risk-holder-nm}
For any $N,M\in\N^{*}$, no matter what the distribution of $W$ is, the $\ell$-type estimator $\widehat f(\bsX)$ taking values in the network class $\cF_{(L,p,K)}$ with
$$p=38(\lfloor\alpha\rfloor+1)^{2}3^{d}d^{\lfloor\alpha\rfloor+1}N\lceil\log_{2}(8N)\rceil,$$
$$L=21(\lfloor\alpha\rfloor+1)^{2}M\lceil\log_{2}(8M)\rceil+2d$$ and a sufficiently large $K$,
satisfies that for any $f\et\in\cH^{\alpha}(\cro{0,1}^{d},B)$ and $n\geq V(\overline\cF_{(L,p)})$, 
\begin{equation}\label{risk-N-M}
\E\cro{\ell(Q_{f\et},Q_{\widehat f})}\leq C_{\epsilon,\sigma,\alpha,d,B}\cro{(NM)^{-2\alpha/d}+\frac{NM}{\sqrt{n}}\left(\log_{2}(2N)\log_{2}(2M)\right)^{3/2}},
\end{equation}
where $C_{\epsilon,\sigma,\alpha,d,B}$ is a constant depending on $\epsilon,\sigma,\alpha,d$ and $B$ only.

In particular, if we take $N=1$ and $M=\lceil n^{d/2(d+2\alpha)}\rceil$, combining Lemma~\ref{l1-ell} with \eref{risk-N-M} allows us to deduce that:
\begin{equation}\label{ell-holder-convege-rate}
\E\cro{\ell(Q_{f\et},Q_{\widehat f})}\leq C_{\epsilon,\sigma,\alpha,d,B}n^{-\frac{\alpha}{d+2\alpha}}\left(\log n\right)^{3/2}.
\end{equation}
For $n$ being sufficiently large such that the right-hand side of \eref{ell-holder-convege-rate} is smaller that 0.78, according to Lemma~\ref{l1-ell}, \eref{ell-holder-convege-rate} is equivalent to 
\[
\E\cro{\|{f\et}-\widehat f\|_{1,P_{W}}}\leq C_{\epsilon,\sigma,\alpha,d,B}n^{-\frac{\alpha}{d+2\alpha}}\left(\log n\right)^{3/2},
\]
which is the typical rate of convergence with respect to the $\L_{1}(P_{W})$-norm.
\end{cor}
The proof of Corollary~\ref{risk-holder-nm} is deferred to Section~\ref{risk-nm-proof}. Our remarks are provided below.
\begin{rmk}
As we will see later it in Theorem~\ref{lower-composite}, the convergence rate $n^{-\alpha/(d+2\alpha)}$ is minimax optimal with respect to the distance $\ell(\cdot,\cdot)$, at least when $W$ is uniformly distributed on $\cro{0,1}^{d}$. Therefore, the risk bound \eref{ell-holder-convege-rate} we obtained is optimal up to a logarithmic factor. As it was shown in Section~7 of \cite{Baraud2021}, $\ell$-estimators are not always optimal when addressing various estimation problems, which differs from $\rho$-estimators. However, by combining the upper bound \eref{ell-holder-convege-rate} and the subsequent lower bound stated in Theorem~\ref{lower-composite}, we demonstrate that implementing the $\ell$-type estimation procedure is optimal within our framework and offers more robustness compared to $\rho$-estimators.
\end{rmk}
\begin{rmk}
A noteworthy aspect of the presented result is that the stochastic error is not dependent on the upper bound of the sup-norms for all functions within the class $\overline\cF_{(L,p,K)}$. This is not the case, for example, in the results established in Lemma~4 of \cite{Schmidt2020} and Theorem~4.2 of \cite{Jiao2023}, both of which analyze the performance of the least squares estimator. As a consequence, the final risk bound they established deteriorates with the enlargement of the model they considered due to the inclusion of such an upper bound in their stochastic error terms. From this perspective, our estimation method does not suffer from this drawback. Therefore, we can accommodate a sufficiently large value of $K$ without compromising the risk bound for the resulting estimator.
\end{rmk}

\section{Circumventing the curse of dimensionality}\label{curse}
As we observed in Section~\ref{neural}, the minimax optimal rate over an $\alpha$-H\"older class on $\sW=\cro{0,1}^{d}$ is of order $n^{-\alpha/(d+2\alpha)}$. This rate slows down significantly as the dimensionality $d$ increases, a phenomenon known as the curse of dimensionality. To overcome this issue, in this section, we introduce structural assumptions on $f\et$ and construct specific models using deep ReLU neural networks to implement our procedure.

One natural structure for the regression function $f\et$ for neural networks to exhibit advantages is a composition of multiple functions, which was previously explored by \cite{Schmidt2020}. More precisely, for any $k\in\N^{*}$, ${\bf{d}}=(d_{0},\ldots,d_{k})\in(\N^{*})^{k+1}$, ${\bf{t}}=(t_{0},\ldots,t_{k})\in(\N^{*})^{k+1}$, ${\bs{\alpha}}=(\alpha_{0},\ldots,\alpha_{k})\in(\R_{+}^{*})^{k+1}$ and a finite constant $B\geq0$, we denote $\cF(k,{\bf{d}},{\bf{t}},{\bs{\alpha}},B)$ the class of functions as, 
\begin{align}
\cF(k,{\bf{d}},{\bf{t}},{\bs{\alpha}},B)=&\left\{f_{k}\circ\cdots\circ f_{0},\;f_{i}=(f_{ij})_{j}:\cro{a_{i},b_{i}}^{d_{i}}\rightarrow\cro{a_{i+1},b_{i+1}}^{d_{i+1}},\right.\nonumber\\
&\quad\quad\left.f_{ij}\in\cH^{\alpha_{i}}(\cro{a_{i},b_{i}}^{t_{i}},B)\;\mbox{and } (|a_{i}|\vee|b_{i}|)\leq B\right\},\label{def-f-composite}
\end{align}
where $a_{0}=0$, $b_{0}=1$, $d_{0}=d$ and $d_{k+1}=1$. In what follows, we assume the existence of an underlying regression function $f\et=f_{k}\circ\cdots\circ f_{0}\in\cF(k,{\bf{d}},{\bf{t}},{\bs{\alpha}},B)$ such that $Q\et=Q_{f\et}$ (or at least $Q\et$ is close to $Q_{f\et}$ with respect to $\ell$), where the values of $k$, ${\bf{d}}$, ${\bf{t}}$ and ${\bs{\alpha}}$ are considered to be known. We will then proceed to construct suitable networks for approximating the class $\cF(k,{\bf{d}},{\bf{t}},{\bs{\alpha}},B)$ and implement the $\ell$-type estimation procedure to derive the final estimator $Q_{\widehat f}$ of $Q\et$.


In such a composition structure, the power of approximation based on the neural network actually relies on the so-called effective smoothness indices, which are defined as $$\alpha^{*}_{i}=\alpha_{i}\prod_{l=i+1}^{k}\left(\alpha_{l}\wedge1\right),\quad\mbox{for }\ i\in\{0,\ldots,k-1\}$$ and $\alpha^{*}_{k}=\alpha_{k}$.  
Based on Proposition~\ref{appro-holder} and the basic operation rules of the neural networks, we establish the following result to approximate any function belonging to $\cF(k,{\bf{d}},{\bf{t}},{\bs{\alpha}},B)$.
\begin{prop}\label{approximate-add}
Assuming that $f\in\cF(k,{\bf{d}},{\bf{t}},{\bs{\alpha}},B)$ with $\cF(k,{\bf{d}},{\bf{t}},{\bs{\alpha}},B)$ defined by \eref{def-f-composite}. For all $i\in\{0,\ldots,k\}$, denote $$p_{i}=114(\lfloor\alpha_{i}\rfloor+1)^{2}3^{t_{i}}t_{i}^{\lfloor\alpha_{i}\rfloor+1}$$ and $$L_{i}=21(\lfloor\alpha_{i}\rfloor+1)^{2}\lceil n^{t_{i}/2(t_{i}+2\alpha_{i}^{*})}\rceil\lceil\log_{2}(8\lceil n^{t_{i}/2(t_{i}+2\alpha_{i}^{*})}\rceil)\rceil+2t_{i}.$$ There exists a function $\overline f$ implemented by a ReLU network with a width of $\overline p=\max_{i\in\{0,\ldots,k-1\}}d_{i+1}p_{i}$ and a depth of $\overline L=k+\sum_{i=0}^{k}L_{i}$ such that
\begin{align*}
&\|f-\overline f\|_{\infty}\\
&\quad\leq (2B)^{1+\sum_{i=1}^{k}\alpha_{i}}\left(\prod_{i=0}^{k}\sqrt{d_{i}}\right)\cro{\sum_{i=0}^{k}C_{\alpha_{i},t_{i},B}^{\prod_{l=i+1}^{k}(\alpha_{l}\wedge1)}(\lceil n^{t_{i}/2(t_{i}+2\alpha_{i}^{*})}\rceil)^{-2\alpha^{*}_{i}/t_{i}}},
\end{align*}
where $$C_{\alpha_{i},t_{i},B}=19(2B)^{\alpha_{i}+1}(\lfloor\alpha_{i}\rfloor+1)^{2}t_{i}^{\lfloor\alpha_{i}\rfloor+(\alpha_{i}\vee1)/2}.$$ 

%
\end{prop}
The proof of Proposition~\ref{approximate-add} is postponed to Section~\ref{approximation-composite-proof}. The presented approximation result is notable for offering a well-defined structure for neural networks to effectively implement diverse estimation approaches, as compared to the sparsity-based networks considered in \cite{Schmidt2020}. From this point of view, Proposition~\ref{approximate-add} is more informative. 
Building upon the results of Theorem~\ref{general}, Proposition~\ref{vc-dim}, \ref{approximate-add}, and Lemma~\ref{network-dense}, we can derive the following risk bound for the $\ell$-type estimators.



\begin{cor}\label{composite-holder-riskbound}
Assume that $f\et\in\cF(k,{\bf{d}},{\bf{t}},{\bs{\alpha}},B)$ with $\cF(k,{\bf{d}},{\bf{t}},{\bs{\alpha}},B)$ defined by \eref{def-f-composite}. For all $i\in\{0,\ldots,k\}$, we set $$p_{i}=114(\lfloor\alpha_{i}\rfloor+1)^{2}3^{t_{i}}t_{i}^{\lfloor\alpha_{i}\rfloor+1}$$ and $$L_{i}=21(\lfloor\alpha_{i}\rfloor+1)^{2}\lceil n^{t_{i}/2(t_{i}+2\alpha_{i}^{*})}\rceil\lceil\log_{2}(8\lceil n^{t_{i}/2(t_{i}+2\alpha_{i}^{*})}\rceil)\rceil+2t_{i}.$$ Whatever the distribution of $W$, any $\ell$-type estimator $\widehat f(\bsX)$ implemented by a ReLU neural network $\cF_{(\overline L,\overline p,K)}$ with 
$$\overline L=k+\sum_{i=0}^{k}L_{i},\quad\quad\overline p=\max_{i=0,\ldots,k}d_{i+1}p_{i}$$ and a sufficiently large $K$,
satisfies that for all $n\geq V(\overline\cF_{(\overline L,\overline p)})$,
\begin{equation}\label{composite-holder-risk}
\E\cro{\ell\pa{Q_{f\et},Q_{\widehat f}}}\leq C_{\epsilon,\sigma,k,{\bf{d}},{\bf{t}},{\bs{\alpha}},B}\left(\sum_{i=0}^{k}n^{-\frac{\alpha_{i}^{*}}{t_{i}+2\alpha_{i}^{*}}}\right)(\log n)^{3/2},
\end{equation}
where $C_{\epsilon,\sigma,k,{\bf{d}},{\bf{t}},{\bs{\alpha}},B}$ is a numerical constant depending on $\epsilon,\sigma,k,{\bf{d}},{\bf{t}},{\bs{\alpha}},B$ only.
\end{cor}
We left the proof of Corollary~\ref{composite-holder-riskbound} to Section~\ref{proof-composite-holder}. Our comments are presented as follows.
\begin{rmk}
Denoting $$\phi_{n}=\max_{i=0,\ldots,k}n^{-\alpha^{*}_{i}/(2\alpha^{*}_{i}+t_{i})},$$ the result \eref{composite-holder-risk} we have established indicates that, up to a logarithmic term, the $\ell$-type estimator $\widehat f$ based on the class $\cF_{(\overline L,\overline p,K)}$ converges to the regression function $f\et$ at the rate of $\phi_{n}$. Furthermore, for sufficiently large $n$ such that the right-hand side of \eref{composite-holder-risk} is smaller than 0.78, upon applying Lemma~\ref{l1-ell}, we obtain 
\[
\E\cro{\|f\et-\widehat f\|_{1,P_{W}}}\leq C_{\epsilon,\sigma,k,{\bf{d}},{\bf{t}},{\bs{\alpha}},B}\phi_{n}(\log n)^{3/2}.
\]
This aligns with the risk bound established in Theorem 1 of \cite{Schmidt2020} for the least squares estimator with respect to the $\L_{2}(P_{W})$-norm. 
\end{rmk}
\begin{rmk}
If the situation deviates from the ideal scenario where $Q\et=Q_{f\et}$ and $f\et\in\cF(k,{\bf{d}},{\bf{t}},{\bs{\alpha}},B)$, a bias term $\inf_{f\in\cF(k,{\bf{d}},{\bf{t}},{\bs{\alpha}},B)}\ell(Q\et,Q_{f})$ will be included in the final risk bound \eref{composite-holder-risk}. However, as long as the bias term is not significantly larger than the quantity on the right-hand side of \eref{composite-holder-risk}, the accuracy of the resulting estimator $\widehat f$ remains on the same order of magnitude as in the ideal case. This follows from the robustness property of the $\ell$-type estimator as we have explained in Section~\ref{ell-construction-section}.


\end{rmk}
The following lower bound demonstrates that the convergence rate $\phi_{n}$ is minimax optimal, at least when $W$ is uniformly distributed on $\cro{0,1}^{d_{0}}$.
\begin{thm}\label{lower-composite}
Let $P_{W}$ be the uniform distribution on $\cro{0,1}^{d_{0}}$. For any $k\in\N^{*}$, ${\bf{d}}\in(\N^{*})^{k+1}$, ${\bf{t}}\in(\N^{*})^{k+1}$ such that $t_{j}\leq\min(d_{0},\ldots,d_{j-1})$ for all $j$, any ${\bm{\alpha}}\in(\R_{+}^{*})^{k+1}$ and $B>0$ large enough, there exists a positive constant $c$ such that
\[
\inf_{\widehat f}\sup_{f\et\in\cF(k,{\bf{d}},{\bf{t}},{\bm{\alpha}},B)}\E\cro{\ell\pa{Q_{f\et},Q_{\widehat f}}}\geq c\phi_{n},
\]
where the infimum runs among all possible estimators of $f\et$.
\end{thm}
The proof of Theorem~\ref{lower-composite} is deferred to Section~\ref{proof-lower-bound}. 
\section{Proofs}\label{proof-paper}
\subsection{Proof of Lemma~\ref{t-property}}\label{Lem-t-property}
\begin{proof}
Drawing on the formulation of $t_{(f_{1},f_{2})}$ and the definition of $\ell(\cdot,\cdot)$ as provided in \eref{n1-ell-def}, we can deduce that
\begin{align*}
&\E_{P\et}\cro{t_{(f_{1},f_{2})}(W,Y)}\\
=&\int_{\sW}\cro{Q\et_{(w)}\left(q_{f_{2}(w)}>q_{f_{1}(w)}\right)-Q_{f_{1}(w)}\left(q_{f_{2}(w)}>q_{f_{1}(w)}\right)}dP_{W}(w)\\
\leq&\int_{\sW}\|Q\et_{(w)}-Q_{f_{1}(w)}\|_{TV}dP_{W}(w)\\
=&\ \ell(Q\et,Q_{f_{1}}),
\end{align*}
which gives the second inequality in \eref{t-preproperty}. Furthermore, for any two probabilities $P$ and $R$ on the measured space $(\sX,\cX)$, it is well known that the total variation distance can equivalently be written as
$$\|P-R\|_{TV}=R(r>p)-P(r>p),$$ where $p$ and $r$ stand for the respective densities of $P$ and $R$ with respect to some common dominating measure $\mu$. Given this fact, we can calculate
\begin{align*}
&\E_{P\et}\cro{t_{(f_{1},f_{2})}(W,Y)}\\
=&\int_{\sW}\cro{Q\et_{(w)}\left(q_{f_{2}(w)}>q_{f_{1}(w)}\right)-Q_{f_{2}(w)}\left(q_{f_{2}(w)}>q_{f_{1}(w)}\right)}dP_{W}(w)\\
\quad&+\int_{\sW}\cro{Q_{f_{2}(w)}\left(q_{f_{2}(w)}>q_{f_{1}(w)}\right)-Q_{f_{1}(w)}\left(q_{f_{2}(w)}>q_{f_{1}(w)}\right)}dP_{W}(w)\\
\geq&\int_{\sW}\cro{\|Q_{f_{1}(w)}-Q_{f_{2}(w)}\|_{TV}-\|Q\et_{(w)}-Q_{f_{2}(w)}\|_{TV}}dP_{W}(w)\\
=&\ \ell(Q_{f_{1}},Q_{f_{2}})-\ell(Q\et,Q_{f_{2}}),
\end{align*}
which yields the first inequality in \eref{t-preproperty}.
\end{proof}

\subsection{Proof of Theorem~\ref{general}}\label{general-thm-proof}
Prior to proving Theorem~\ref{general}, we will initially establish several auxiliary results that will serve as the foundation for deriving Theorem~\ref{general}.
\begin{prop}\label{vcsubsets}
For any $\overline f\in\cF$, we define $$\sC_{+}(\cF,\overline f)=\left\{\left\{(w,y)\in\sW\times\sY\ s.t.\ q_{f(w)}(y)>q_{\overline f(w)}(y)\right\},\;f\in\cF\right\}$$ and $$\sC_{-}(\cF,\overline f)=\left\{\left\{(w,y)\in\sW\times\sY\ s.t.\ q_{f(w)}(y)<q_{\overline f(w)}(y)\right\},\;f\in\cF\right\}.$$ Under Assumption~\ref{vcass}, the classes of subsets $\sC_{+}(\cF,\overline f)$ and $\sC_{-}(\cF,\overline f)$ are both VC with dimensions not larger than $9.41V$. 
\end{prop}
\begin{proof}
We first prove the result holds for the class $\sC_{+}(\cF,\overline f)$. For any $f\in\cF$, we define the function $\widetilde q_{f}$ on $\sW\times\sY$ as $$\widetilde q_{f(w)}(y)=\exp\cro{\frac{yf(w)}{\sigma^{2}}-\frac{f^{2}(w)}{2\sigma^{2}}}.$$ 
Then the class of subsets $\sC_{+}(\cF,\overline f)$ can be rewritten as $$\sC_{+}(\cF,\overline f)=\left\{\left\{(w,y)\in\sW\times\sY\ s.t.\ \widetilde q_{f(w)}(y)>\widetilde q_{\overline f(w)}(y)\right\},\;f\in\cF\right\}.$$ 

We introduce the result of Proposition~5 in \cite{Baraud2020} as follows.
\begin{prop}\label{vcconnect}
Let $I\subset\R$ be a non-trivial interval and $\cF$ a class of functions from $\sW$ into $I$. If $\cF$ is VC-subgraph on $\sW$ with dimension not larger than $V$, the class of functions $$\left\{h_{f}:(w,y)\mapsto e^{S(y)f(w)-A(f(w))},\;f\in\cF\right\}$$ is VC-subgraph on $\sW\times\sY$ with dimension not larger than $9.41V$, where $S$ is a real-valued measurable function on $\sY$ and $A$ is convex and continuous on $I$.
\end{prop}
Note that function $\widetilde q_{f(w)}(y)$ takes a particular form as described in Proposition~\ref{vcconnect} with $S(y)=y/\sigma^{2}$, for all $y\in\R$, and $A(u)=u^{2}/(2\sigma^{2})$, for all $u\in\R$. Therefore, under Assumption~\ref{vcass}, the class of functions $\left\{\widetilde q_{f},\;f\in\cF\right\}$ on $\sW\times\sY$ is VC-subgraph with dimension not larger than $9.41V$. Moreover, since for any given function $\overline f\in\cF$, $\widetilde q_{\overline f}$ is a fixed function taking its values in $\R$, applying Lemma 2.6.18 (v) of \cite{MR1385671} (see also Proposition~42 (i) in \cite{MR3595933}), we obtain that the class of functions $\left\{\widetilde q_{f}-\widetilde q_{\overline f},\;f\in\cF\right\}$ on $\sW\times\sY$ is VC-subgraph with dimension not larger than $9.41V$. According Proposition~2.1 of \cite{Baraud2016}, $\left\{\widetilde q_{f}-\widetilde q_{\overline f},\;f\in\cF\right\}$ is weak VC-major with dimension not larger than $9.41V$, which implies that the class of subsets $$\left\{\left\{(w,y)\in\sW\times\sY\ s.t.\ \widetilde q_{f(w)}(y)-\widetilde q_{\overline f(w)}(y)>0\right\},\;f\in\cF\right\}$$ is a VC-class of subsets of $\sW\times\sY$ with dimension not larger than $9.41V$. Hence, the conclusion holds for $\sC_{+}(\cF,\overline f)$.

Now we show the conclusion also holds for the class of subsets $\sC_{-}(\cF,\overline f)$. As we have seen, under Assumption~\ref{vcass}, $\left\{\widetilde q_{f},\;f\in\cF\right\}$ on $\sW\times\sY$ is VC-subgraph with dimension not larger than $9.41V$. By applying Proposition~42 (iii) of \cite{MR3595933}, we can establish that $\{-\widetilde q_{f},\;f\in\cF\}$ on $\sW\times\sY$ is a VC-subgraph with dimension not exceeding $9.41V$. As a result of Lemma 2.6.18 (v) of \cite{MR1385671}, this property also holds for the class $\{\widetilde q_{\overline f}-\widetilde q_{f},\;f\in\cF\}$ for any fixed $\overline f\in\cF$. Finally, we can conclude using a similar argument as we did for $\sC_{+}(\cF,\overline f)$. 
\end{proof}

\begin{prop}\label{vcequa}
For any $\overline f\in\cF$, we define $$\sC_{=}(\cF,\overline f)=\left\{\{w\in\sW\ s.t.\ f(w)=\overline f(w)\},\;f\in\cF\right\}.$$ Under Assumption~\ref{vcass}, the class of subsets $\sC_{=}(\cF,\overline f)$ is a VC-class of sets on $\sW$ with dimension not larger than $9.41V$. 
\end{prop}
\begin{proof}
We set $$\sC_{\geq}(\cF,\overline f)=\left\{\{w\in\sW\ s.t.\ f(w)-\overline f(w)\geq0\},\;f\in\cF\right\}$$ and $$\sC_{\leq}(\cF,\overline f)=\left\{\{w\in\sW\ s.t.\ f(w)-\overline f(w)\leq0\},\;f\in\cF\right\}.$$ Since $\cF$ is VC-subgraph on $\sW$ with dimension not larger than $V$ and $\overline f$ is a fixed function, $\{f-\overline f,\; f\in\cF\}$ is VC-subgraph on $\sW$ with dimension not larger than $V$ as a consequence of applying Proposition~42 (i) in \cite{MR3595933}. According to Proposition~2.1 of \cite{Baraud2016}, $\{f-\overline f,\; f\in\cF\}$ is weak VC-major with dimension not larger than $V$, which implies that the class of subsets $$\left\{\left\{w\in\sW\ s.t.\ f(w)>\overline f(w)\right\},\;f\in\cF\right\}$$ is a VC-class of subsets of $\sW$ with dimension not larger than $V$. Then Lemma~2.6.17 (i) of \cite{MR1385671} implies that $\sC_{\leq}(\cF,\overline f)$ is a VC-class of subsets of $\sW$ with dimension not larger than $V$. Following a similar argument, we can show that the same conclusion also holds for the class $\sC_{\geq}(\cF,\overline f)$.

Writing $$\sC_{\geq}(\cF,\overline f)\bigwedge\sC_{\leq}(\cF,\overline f)=\left\{C_{\geq}\cap C_{\leq},\;C_{\geq}\in\sC_{\geq}(\cF,\overline f),\;C_{\leq}\in\sC_{\leq}(\cF,\overline f)\right\},$$ we can deduce that $\sC_{\geq}(\cF,\overline f)\bigwedge\sC_{\leq}(\cF,\overline f)$ is a VC-class of subsets of $\sW$ with dimension not larger than $9.41V$ according to Theorem~1.1 of \cite{MR2797943}. It is easy to note that $\sC_{=}(\cF,\overline f)\subset\sC_{\geq}(\cF,\overline f)\bigwedge\sC_{\leq}(\cF,\overline f)$, which implies the completion of the proof.
\end{proof}

\begin{lem}\label{thmlemma}
Under Assumption~\ref{vcass}, whatever the conditional distributions $\gQ\et=(Q\et_{1},\ldots,Q\et_{n})$ of the $Y_{i}$ given $W_{i}$ and the distributions of $W_{i}$, any $\ell$-type estimator $\widehat f$ based on the set $\cF$ satisfies that for any $\overline f\in\cF$ and any $\xi>0$, with a probability at least $1-e^{-\xi}$,
\begin{equation}\label{equa4}
\ell(\gQ_{\overline f},\gQ_{\widehat f})\leq2\ell(\gQ\et,\gQ_{\overline f})+\frac{2\bs\vartheta(\overline f)}{n}+\sqrt{\frac{8(\xi+\log2)}{n}}+\frac{\epsilon}{n},
\end{equation}
where 
\begin{align*}
\bs\vartheta(\overline f)=&\E\cro{\sup_{f'\in\cF}\cro{\gT_{l}(\bsX,\overline f,f')-\E\cro{\gT_{l}(\bsX,\overline f,f')}}}\\
&\vee\E\cro{\sup_{f'\in\cF}\cro{\E\cro{\gT_{l}(\bsX,f',\overline f)}-\gT_{l}(\bsX,f',\overline f)}}.
\end{align*}

\end{lem}

\begin{proof}
The proof of Lemma~\ref{thmlemma} builds upon the idea presented in the proof of Theorem~1 in \cite{Baraud2022}, but with certain modifications to adapt it to the regression setting.

For any $f_{1}, f_{2}\in\cF$, define
$$\gZ_{+}(\bsX,f_{1})=\sup_{f_{2}\in\cF}\cro{\gT_{l}(\bsX,f_{1},f_{2})-\E\cro{\gT_{l}(\bsX,f_{1},f_{2})}}$$
$$\gZ_{-}(\bsX,f_{1})=\sup_{f_{2}\in\cF}\cro{\E\cro{\gT_{l}(\bsX,f_{2},f_{1})}-\gT_{l}(\bsX,f_{2},f_{1})}$$ and set
$$\gZ(\bsX,f_{1})=\gZ_{+}(\bsX,f_{1})\vee\gZ_{-}(\bsX,f_{1}).$$
As per Lemma~\ref{t-property}, for any $f,\overline f\in\cF$, it holds that
\begin{align}
n\ell(\gQ_{\overline f},\gQ_{f})&\leq n\ell(\gQ\et,\gQ_{\overline f})+\E\cro{\gT_{l}(\bsX,f,\overline f)}\nonumber\\
&\leq n\ell(\gQ\et,\gQ_{\overline f})+\E\cro{\gT_{l}(\bsX,f,\overline f)}-\gT_{l}(\bsX,f,\overline f)+\gT_{l}(\bsX,f,\overline f)\nonumber\\
&\leq n\ell(\gQ\et,\gQ_{\overline f})+\gZ(\bsX,\overline f)+\gT_{l}(\bsX,f,\overline f)\nonumber\\
&\leq n\ell(\gQ\et,\gQ_{\overline f})+\gZ(\bsX,\overline f)+\gT_{l}(\bsX,f).\label{basic-inequa}
\end{align}

By utilizing \eref{basic-inequa}, substituting $f$ with $\widehat f(\bsX)\in\sE(\bsX,\epsilon)$, and employing the definition of $\widehat f(\bsX)$, we can derive that
\begin{align}
n\ell(\gQ_{\overline f},\gQ_{\widehat f})&\leq n\ell(\gQ\et,\gQ_{\overline f})+\gZ(\bsX,\overline f)+\gT_{l}(\bsX,\widehat f)\nonumber\\
&\leq n\ell(\gQ\et,\gQ_{\overline f})+\gZ(\bsX,\overline f)+\gT_{l}(\bsX,\overline f)+\epsilon.\label{equa1}
\end{align}
Moreover, we can compute that
\begin{align}
\gT_{l}(\bsX,\overline f)&=\sup_{f\in\cF}\gT_{l}(\bsX,\overline f,f)\nonumber\\
&\leq\sup_{f\in\cF}\cro{\gT_{l}(\bsX,\overline f,f)-\E\cro{\gT_{l}(\bsX,\overline f,f)}}+\sup_{f\in\cF}\E\cro{\gT_{l}(\bsX,\overline f,f)}\nonumber\\
&\leq\gZ(\bsX,\overline f)+n\ell(\gQ\et,\gQ_{\overline f}),\label{gT-inequa}
\end{align}
where the second inequality is obtained by applying Lemma~\ref{t-property}. Combining \eref{equa1} and \eref{gT-inequa}, we obtain that for any $\overline f\in\cF$,
\begin{equation}\label{inequa-semi}
n\ell(\gQ_{\overline f},\gQ_{\widehat f})\leq2\gZ(\bsX,\overline f)+2n\ell(\gQ\et,\gQ_{\overline f})+\epsilon.
\end{equation}
In what follows, we study the term $\gZ(\bsX,\overline f)$ to have a further insight of the risk bound for the estimator $\widehat f$. 
It is worth noting that for any $\overline f, f\in\cF$ and $(w,y),(w',y')\in\sW\times\sY$, the following inequality holds:
$$\big|t_{(\overline f,f)}(w,y)-t_{(\overline f,f)}(w',y')\big|\leq2.$$ Writing $\bs{x}=(x_{1},\ldots,x_{n})\in\sX^{n}$ and $\bs{x}'_{(i)}=(x_{1},\ldots,x'_{i},\ldots,x_{n})\in\sX^{n}$, as an immediate consequence, we can derive that
$$\frac{1}{2}\big|\gZ_{+}(\bs{x},\overline f)-\gZ_{+}(\bs{x}'_{(i)},\overline f)\big|\leq1.$$
By following a similar approach as in the proof of Lemma~2 in \cite{Baraud2021} and considering the term $\xi$ replaced with $\xi+\log2$, one can conclude that with a probability of at least $1-(1/2)e^{-\xi}$,
\begin{align}
\gZ_{+}(\bsX,\overline f)&\leq\E\cro{\gZ_{+}(\bsX,\overline f)}+\sqrt{2n(\xi+\log2)}\nonumber\\
&=\E\cro{\sup_{f\in\cF}\cro{\gT_{l}(\bsX,\overline f,f)-\E\cro{\gT_{l}(\bsX,\overline f,f)}}}+\sqrt{2n(\xi+\log2)}\nonumber\\
&\leq\bs\vartheta(\overline f)+\sqrt{2n(\xi+\log2)}.\label{equa2}
\end{align}
A similar argument gives that with a probability at least $1-(1/2)e^{-\xi}$,
\begin{equation}\label{equa3}
\gZ_{-}(\bsX,\overline f)\leq\bs\vartheta(\overline f)+\sqrt{2n(\xi+\log2)}.
\end{equation}
By combining \eref{equa2} and \eref{equa3}, we can derive that with a probability at least $1-e^{-\xi}$,
\begin{equation}\label{gZ-bound}
\gZ(\bsX,\overline f)=\gZ_{+}(\bsX,\overline f)\vee\gZ_{-}(\bsX,\overline f)\leq\bs\vartheta(\overline f)+\sqrt{2n(\xi+\log2)}.
\end{equation}
Finally, plugging \eref{gZ-bound} into \eref{inequa-semi} gives the upper bound
\begin{equation}
\ell(\gQ_{\overline f},\gQ_{\widehat f})\leq2\ell(\gQ\et,\gQ_{\overline f})+\frac{2\bs\vartheta(\overline f)}{n}+\sqrt{\frac{8(\xi+\log2)}{n}}+\frac{\epsilon}{n}.
\end{equation}
\end{proof}

\begin{prop}\label{Q-equality}
Let $f_{1}$ and $f_{2}$ be two functions belonging to $\cF$. For all $w\in\sW$, the following equality holds 
\[
Q_{f_{1}(w)}(q_{f_{2}(w)}>q_{f_{1}(w)})=\Phi\left(-\frac{|f_{1}(w)-f_{2}(w)|}{2\sigma}\right)-\frac{1}{2}\1_{f_{1}(w)=f_{2}(w)},
\]
where $\Phi$ stands for the cumulative distribution function of the standard normal distribution.
\end{prop}
\begin{proof}
For all $w\in\sW$ satisfying $f_{1}(w)=f_{2}(w)$, it is easy to see that $Q_{f_{1}(w)}(q_{f_{2}(w)}>q_{f_{1}(w)})=0$. The equality naturally holds since $$\Phi\left(-\frac{|f_{1}(w)-f_{2}(w)|}{2\sigma}\right)-\frac{1}{2}\1_{f_{1}(w)=f_{2}(w)}=\Phi(0)-\frac{1}{2}=0.$$ For all $w\in\sW$ satisfying $f_{1}(w)>f_{2}(w)$,
\begin{align*}
Q_{f_{1}(w)}(q_{f_{2}(w)}>q_{f_{1}(w)})&=\int_{-\infty}^{\cro{f_{1}(w)+f_{2}(w)}/2}q_{f_{1}(w)}(y)dy\\
&=\int_{-\infty}^{\cro{f_{2}(w)-f_{1}(w)}/2\sigma}\frac{1}{\sqrt{2\pi}}e^{-\frac{t^{2}}{2}}dt\\
&=\Phi\left(-\frac{|f_{1}(w)-f_{2}(w)|}{2\sigma}\right).
\end{align*}
For all $w\in\sW$ satisfying $f_{1}(w)<f_{2}(w)$,
\begin{align*}
Q_{f_{1}(w)}(q_{f_{2}(w)}>q_{f_{1}(w)})&=\int_{\cro{f_{1}(w)+f_{2}(w)}/2}^{+\infty}q_{f_{1}(w)}(y)dy\\
&=\int_{\cro{f_{2}(w)-f_{1}(w)}/2\sigma}^{+\infty}\frac{1}{\sqrt{2\pi}}e^{-\frac{t^{2}}{2}}dt\\
&=1-\Phi\left(\frac{f_{2}(w)-f_{1}(w)}{2\sigma}\right)\\
&=\Phi\left(-\frac{|f_{1}(w)-f_{2}(w)|}{2\sigma}\right).
\end{align*}
Therefore, we can conclude the equality.
\end{proof}

The following result comes from the Proposition~3.1 in \cite{Baraud2016}, and we shall repeatedly use it in our proof.
\begin{lem}\label{Baraud2016prop}
Let $X_{1},\ldots,X_{n}$ be independent random variables with values in $(E,\cE)$ and $\cC$ a $VC$-class of subsets of $E$ with $VC$-dimension not larger than $V\geq1$ that satisfies for $\sigma\in(0,1]$, $\sum_{i=1}^{n}\P(X_{i}\in C)\leq n\sigma^{2}$ for all $C\in\cC$. Then,
\[
\E\cro{\sup_{C\in\cC}\Big|\sum_{i=1}^{n}\left(\1_{C}(X_{i})-\P(X_{i}\in C)\right)\Big|}\leq 10(\sigma\vee a)\sqrt{nV\cro{5+\log\left(\frac{1}{\sigma\vee a}\right)}} 
\]
where $$a=\cro{32\sqrt{\frac{(V\wedge n)}{n}\log\left(\frac{2en}{V\wedge n}\right)}}\wedge1.$$
\end{lem}

To prove Theorem~\ref{general}, we also need the following result, which can be obtained by making a modification to the proof of Theorem~2 in \cite{Baraud2020}.
\begin{lem}\label{supemp}
Let $W_{1},\ldots,W_{n}$ be $n$ independent random variables with values in $(\sW,\cW)$ and $\cF$ an at most countable VC-subgraph class of functions with values in $[0,1]$ and VC-dimension not larger than $V\geq 1$. If 
\[
Z(\cF)=\sup \limits_{f\in\cF} \left|\sum \limits_{i=1}^{n}(f(W_{i})-\E\cro{f(W_{i})})\right|\;\; \text{and}\;\; \sup_{f\in\cF}\frac{1}{n}\sum_{i=1}^{n}\E\cro{f^{2}(W_{i})}\leq\sigma^{2}\leq 1,
\]
then
\begin{equation*}
\E\cro{Z(\cF)}\leq 4.61\sqrt{nV\sigma^{2}\cL(\sigma)}+85V\cL(\sigma),
\end{equation*}
with $\cL(\sigma)=9.11+\log(1/\sigma^{2})$.
\end{lem}

\begin{proof}[Proof of Theorem~\ref{general}]
Now we will proceed to prove Theorem~\ref{general}. Utilizing the result from Lemma~\ref{thmlemma}, we only need to establish an upper bound for the term $\bs\vartheta(\overline f)$. Let us express $\bs\vartheta(\overline f)$ as $\bs\vartheta(\overline f)=\bs\vartheta_{1}(\overline f)\vee\bs\vartheta_{2}(\overline f)$, where $$\bs\vartheta_{1}(\overline f)=\E\cro{\sup_{f'\in\cF}\cro{\gT_{l}(\bsX,\overline f,f')-\E\cro{\gT_{l}(\bsX,\overline f,f')}}}$$ and $$\bs\vartheta_{2}(\overline f)=\E\cro{\sup_{f'\in\cF}\cro{\E\cro{\gT_{l}(\bsX,f',\overline f)}-\gT_{l}(\bsX,f',\overline f)}}.$$ In what follows, we will derive an upper bound for the term $\bs\vartheta_{1}(\overline f)$.

For any $f_{1},f_{2}\in\cF$, define $$g_{(f_{1},f_{2})}(w,y)=\1_{q_{f_{2}(w)}(y)>q_{f_{1}(w)}(y)},\quad\mbox{for all}\ (w,y)\in\sW\times\sY.$$
Let $\Phi$ be the cumulative distribution function of the standard normal distribution. For any $f_{1},f_{2}\in\cF$, define $$h_{(f_{1},f_{2})}(w)=\Phi\left(-\frac{|f_{1}(w)-f_{2}(w)|}{2\sigma}\right),\quad\mbox{for all}\ w\in\sW$$ and $$k_{(f_{1},f_{2})}(w)=\frac{1}{2}\1_{f_{1}(w)=f_{2}(w)},\quad\mbox{for all}\ w\in\sW.$$ Given any $f_{1},f_{2}\in\cF$, according to Proposition~\ref{Q-equality}, we have that for all $(w,y)\in\sW\times\sY$,
\begin{align}
t_{(f_{1},f_{2})}(w,y)&=g_{(f_{1},f_{2})}(w,y)-\cro{h_{(f_{1},f_{2})}(w)-k_{(f_{1},f_{2})}(w)}\nonumber\\
&=g_{(f_{1},f_{2})}(w,y)-h_{(f_{1},f_{2})}(w)+k_{(f_{1},f_{2})}(w).\label{qequa}
\end{align}
By the definition of $\gT_{l}$ and the equality \eref{qequa}, we deduce that
\begin{align*}
\bs\vartheta_{1}(\overline f)=&\ \E\cro{\sup_{f'\in\cF}\cro{\sum_{i=1}^{n}\left(t_{(\overline f,f')}(W_{i},Y_{i})-\E\cro{t_{(\overline f,f')}(W_{i},Y_{i})}\right)}}\\
\leq&\ \E\cro{\sup_{f'\in\cF}\left|\sum_{i=1}^{n}\left(g_{(\overline f,f')}(W_{i},Y_{i})-\E\cro{g_{(\overline f,f')}(W_{i},Y_{i})}\right)\right|}\\
&+\E\cro{\sup_{f'\in\cF}\left|\sum_{i=1}^{n}\left(h_{(\overline f,f')}(W_{i})-\E\cro{h_{(\overline f,f')}(W_{i})}\right)\right|}\\
&+\E\cro{\sup_{f'\in\cF}\left|\sum_{i=1}^{n}\left(k_{(\overline f,f')}(W_{i})-\E\cro{k_{(\overline f,f')}(W_{i})}\right)\right|}.
\end{align*}
As it has been shown in Proposition~\ref{vcsubsets} that under Assumption~\ref{vcass}, the class of subset $\sC_{+}(\cF,\overline f)$ is VC with dimension not larger than $9.41V$. Hence, applying Lemma~\ref{Baraud2016prop} with $\sigma=1$, we can obtain that
\begin{equation}\label{bound1}
\E\cro{\sup_{f'\in\cF}\left|\sum_{i=1}^{n}\left(g_{(\overline f,f')}(W_{i},Y_{i})-\E\cro{g_{(\overline f,f')}(W_{i},Y_{i})}\right)\right|}\leq68.6\sqrt{nV}.
\end{equation}
According to Proposition~\ref{vcequa}, under Assumption~\ref{vcass}, the class of subsets $\sC_{=}(\cF,\overline f)$ is VC on $\sW$ with dimension not larger than $9.41V$. Applying Lemma~\ref{Baraud2016prop} again, we derive that
\begin{equation}\label{bound2}
\E\cro{\sup_{f'\in\cF}\left|\sum_{i=1}^{n}\left(k_{(\overline f,f')}(W_{i})-\E\cro{k_{(\overline f,f')}(W_{i})}\right)\right|}\leq34.3\sqrt{nV}.
\end{equation}
Moreover, under Assumption~\ref{vcass}, the class of functions $\{f'-\overline f,\;f'\in\cF\}$ is VC-subgraph on $\sW$ with dimension not larger than $V$. Given the value of $\sigma>0$, since the function $\psi(z)=-|z|/2\sigma$, for all $z\in\R$ is unimodal, the class $\{\psi\circ(f'-\overline f),\;f'\in\cF\}$ is VC-subgraph on $\sW$ with dimension not larger than $9.41V$, as stated in Proposition~42 (vi) of \cite{MR3595933}. Then according to Proposition~42 (ii) of \cite{MR3595933}, $\{h_{(\overline f,f')},\;f'\in\cF\}=\{\Phi\circ\cro{\psi\circ(f'-\overline f)},\;f'\in\cF\}$ is VC-subgraph on $\sW$ with dimension not larger than $9.41V$. 

It is easy to note that for any $f'\in\cF$, $$\frac{1}{n}\sum_{i=1}^{n}\E\cro{h_{(\overline f,f')}^{2}(W_{i})}\leq 1.$$ Applying Lemma~\ref{supemp} to the class $\{h_{(\overline f,f')},\;f'\in\cF\}$ gives the result that
\begin{equation}
\E\cro{\sup_{f'\in\cF}\left|\sum_{i=1}^{n}\left(h_{(\overline f,f')}(W_{i})-\E\cro{h_{(\overline f,f')}(W_{i})}\right)\right|}\leq42.7\sqrt{nV}+7286.7V.\label{bound3}
\end{equation}
Combining \eref{bound1}, \eref{bound2} and \eref{bound3} together, we can conclude that for any $\overline f\in\cF$
\begin{equation}\label{v1}
\bs\vartheta_{1}(\overline f)\leq145.6\sqrt{nV}+7286.7V.
\end{equation}
By following a similar line of proof, one can also derive that
\begin{equation}\label{v2}
\bs\vartheta_{2}(\overline f)\leq145.6\sqrt{nV}+7286.7V.
\end{equation}
Therefore, \eref{v1} and \eref{v2} together imply that for any $\overline f\in\cF$,
\begin{equation}\label{vbound}
\bs\vartheta(\overline f)=\bs\vartheta_{1}(\overline f)\vee\bs\vartheta_{2}(\overline f)\leq145.6\sqrt{nV}+7286.7V.
\end{equation}
By substituting the bound \eref{vbound} into equation \eref{equa4}, we infer that for any $\overline f\in\cF$ and any $\xi>0$, with a probability at least $1-e^{-\xi}$,
\begin{equation}
\ell(\gQ_{\overline f},\gQ_{\widehat f})\leq2\ell(\gQ\et,\gQ_{\overline f})+291.2\sqrt{\frac{V}{n}}+14573.4\frac{V}{n}+\sqrt{\frac{8(\xi+\log2)}{n}}+\frac{\epsilon}{n},
\end{equation}
which concludes the inequality \eref{exp-inequa-everypoint}. Using the triangle inequality,
\[
\ell(\gQ\et,\gQ_{\widehat f})\leq\ell(\gQ\et,\gQ_{\overline f})+\ell(\gQ_{\overline f},\gQ_{\widehat f}),
\]
we derive that any $\ell$-type estimator $\widehat f$ on the set $\cF$ satisfies that for all $\xi>0$, with a probability at least $1-e^{-\xi}$,
\[
\ell(\gQ\et,\gQ_{\widehat f})\leq3\ell(\gQ\et,\sbQ)+291.2\sqrt{\frac{V}{n}}+14573.4\frac{V}{n}+\sqrt{\frac{8(\xi+\log2)}{n}}+\frac{\epsilon}{n}.
\]
\end{proof}

\subsection{Proof of Lemma~\ref{network-dense}}\label{network-dense-supnorm}
\begin{proof}
Lemma~\ref{network-dense} can be proven using a similar argument as in the proof of Lemma~11 in \cite{ChenModSel}, where the main idea is inspired by the proof of Lemma~5 of \cite{Schmidt2020}. We only need to show that for any $f\in\overline\cF_{(L,{\bs{p}},K)}$, there exists a sequence of functions $f_{i}\in\cF_{(L,{\bs{p}},K)}$, $i\in\N^{*}$ such that $$\lim_{i\rightarrow+\infty}\|f-f_{i}\|_{\infty}=0.$$ For any $f\in\overline\cF_{(L,{\bs{p}},K)}$, recall that it can be written as 
$$f({\bs{w}})=M_{L}\circ\sigma\circ M_{L-1}\circ\cdots\circ\sigma\circ M_{0}({\bs{w}})\quad\mbox{for any\ }{\bs{w}}\in\cro{0,1}^{d},$$ 
where $$M_{l}({\bs{y}})=A_{l}({\bs{y}})+b_{l},\mbox{\quad for\ }l=0,\ldots,L,$$ $A_{l}$ is a $p_{l+1}\times p_{l}$ weight matrix and the shift vector $b_{l}$ is of size $p_{l+1}$ for any $l\in\{0,\ldots,L\}$. 

For $l\in\{1,\ldots,L\}$, we define the function $f_{l}^{+}:\cro{0,1}^{d}\rightarrow\R^{p_{l}}$,
\[
f_{l}^{+}({\bs{w}})=\sigma\circ M_{l-1}\circ\cdots\circ\sigma\circ M_{0}({\bs{w}})
\]
and for $l\in\{1,\ldots,L+1\}$, we define $f_{l}^{-}:\R^{p_{l-1}}\rightarrow\R$
\[
f_{l}^{-}({\bs{x}})=M_{L}\circ\sigma\circ\cdots\circ\sigma\circ M_{l-1}({\bs{x}}).
\]
We set the notations $f_{0}^{+}({\bs{x}})=f_{L+2}^{-}({\bs{x}})={\bs{x}}$. Given a vector ${\bs{v}}=(v_{1},\ldots,v_{p})^{\top}$ of any size $p\in\N^{*}$, we denote $|{\bs{v}}|_{\infty}=\max_{i=1,\ldots,p}|v_{i}|$. 

For any $f\in\overline\cF_{(L,{\bs{p}},K)}$, with the fact that the absolute values of all the parameters are bounded by $K$ and ${\bs{w}}\in\cro{0,1}^{d}$, we have for all $l\in\{1,\ldots,L\}$
\begin{equation*}
\left|f_{l}^{+}({\bs{w}})\right|_{\infty}\leq K_{+}^{l}\prod_{k=0}^{l-1}(p_{k}+1),
\end{equation*}
where $K_{+}=\max\{K,1\}$, and $f_{l}^{-}$, $l\in\{1,\ldots,L+1\}$, is a multivariate Lipschitz function with Lipschitz constant bounded by $\prod_{k=l-1}^{L}(K_{+}p_{k})$.

For any $f\in\overline\cF_{(L,{\bs{p}},K)}$ with weight matrices and shift vectors $\{M_{l}=(A_{l},b_{l})\}_{l=0}^{L}$ and for all $\epsilon>0$, since $\Q$ is dense in $\R$, there exist a $N_{\epsilon}>0$ such that for all $i\geq N_{\epsilon}$, all the non-zero parameters in $f_{i}\in\cF_{(L,{\bs{p}},K)}$ are smaller than $$\frac{\epsilon}{(L+1)\prod_{k=0}^{L+1}\cro{K_{+}(p_{k}+1)}}$$ away from the corresponding ones in $f$. We denote the weight matrices and shift vectors of function $f_{i}$ as $\{M_{l}^{i}=(A_{l}^{i},b_{l}^{i})\}_{l=0}^{L}$. We note that $$f_{i}({\bs{w}})=f_{i,2}^{-}\circ\sigma\circ M^{i}_{0}\circ f_{0}^{+}({\bs{w}})$$ and $$f({\bs{w}})=f_{i,L+2}^{-}\circ M_{L}\circ f_{L}^{+}({\bs{w}}).$$ Therefore, for all $i\geq N_{\epsilon}$ and all ${\bs{w}}\in\cro{0,1}^{d}$
\begin{align*}
\left|f_{i}({\bs{w}})-f({\bs{w}})\right|\leq&\sum_{l=1}^{L}\left|f_{i,l+1}^{-}\circ\sigma\circ M^{i}_{l-1}\circ f_{l-1}^{+}({\bs{w}})-f_{i,l+1}^{-}\circ\sigma\circ M_{l-1}\circ f_{l-1}^{+}({\bs{w}})\right|\\
&+\left|M_{L}^{i}\circ f_{L}^{+}({\bs{w}})-M_{L}\circ f_{L}^{+}({\bs{w}})\right|\\
\leq&\sum_{l=1}^{L}\cro{\prod_{k=l}^{L}K_{+}p_{k}}\left|M^{i}_{l-1}\circ f_{l-1}^{+}({\bs{w}})-M_{l-1}\circ f_{l-1}^{+}({\bs{w}})\right|_{\infty}\\
&+\left|M_{L}^{i}\circ f_{L}^{+}({\bs{w}})-M_{L}\circ f_{L}^{+}({\bs{w}})\right|\\
\leq&\sum_{l=1}^{L+1}\cro{\prod_{k=l}^{L+1}K_{+}p_{k}}\left|M^{i}_{l-1}\circ f_{l-1}^{+}({\bs{w}})-M_{l-1}\circ f_{l-1}^{+}({\bs{w}})\right|_{\infty}\\
\leq&\sum_{l=1}^{L+1}\cro{\prod_{k=l}^{L+1}K_{+}p_{k}}\cro{\left|\left(A^{i}_{l-1}-A_{l-1}\right)\circ f_{l-1}^{+}({\bs{w}})\right|_{\infty}+|b^{i}_{l-1}-b_{l-1}|_{\infty}}\\
<&\frac{\sum_{l=1}^{L+1}\cro{\prod_{k=l}^{L+1}K_{+}p_{k}}\left(p_{l-1}\left|f_{l-1}^{+}({\bs{w}})\right|_{\infty}+1\right)}{(L+1)\prod_{k=0}^{L+1}\cro{K_{+}(p_{k}+1)}}\epsilon\\
<&\epsilon.
\end{align*}
Hence, by the definition we can conclude that $\cF_{(L,{\bs{p}},K)}$ is dense in $\overline\cF_{(L,{\bs{p}},K)}$ with respect to the supremum norm $\|\cdot\|_{\infty}$.

\end{proof}

\subsection{Proof of Corollary~\ref{risk-holder-nm}}\label{risk-nm-proof}
\begin{proof}
Recall that, in accordance with the general result \eref{riskbound}, for any $n\geq V(\overline\cF_{(L,p)})\geq V(\overline\cF_{(L,p,K)})$, we can obtain
\begin{equation}\label{risk-holder-initial}
\E\cro{\ell(Q_{f\et},Q_{\widehat f})}\leq C_{\epsilon}\cro{\inf_{f\in\cF_{(L,p,K)}}\ell(Q_{f\et},Q_{f})+\sqrt{\frac{V(\overline\cF_{(L,p,K)})}{n}}},
\end{equation}
where $C_{\epsilon}>0$ is a numerical constant depending on $\epsilon$ only. Then, applying Lemma~\ref{network-dense} and inequality \eref{tvconnect}, we derive from \eref{risk-holder-initial} that
\begin{align}
\E\cro{\ell(Q_{f\et},Q_{\widehat f})}&\leq C_{\epsilon,\sigma}\cro{\inf_{f\in\cF_{(L,p,K)}}\|f\et-f\|_{1,P_{W}}+\sqrt{\frac{V(\overline\cF_{(L,p,K)})}{n}}}\nonumber\\
&\leq C_{\epsilon,\sigma}\cro{\inf_{f\in\overline\cF_{(L,p,K)}}\|f\et-f\|_{1,P_{W}}+\sqrt{\frac{V(\overline\cF_{(L,p,K)})}{n}}},\label{risk-raw}
\end{align}
where $C_{\epsilon,\sigma}$ is a numerical constant depending only on $\epsilon$ and $\sigma$. On the one hand, as a consequence of Proposition~\ref{appro-holder}, we have that for the network $\overline\cF_{(L,p,K)}$ with 
\begin{equation}\label{p-value}
p=38(\lfloor\alpha\rfloor+1)^{2}3^{d}d^{\lfloor\alpha\rfloor+1}N\lceil\log_{2}(8N)\rceil,
\end{equation}
\begin{equation}\label{L-value}
L=21(\lfloor\alpha\rfloor+1)^{2}M\lceil\log_{2}(8M)\rceil+2d
\end{equation}
and $K$ being large enough,
\begin{align}\label{risk-appro-bound}
\inf_{f\in\overline\cF_{(L,p,K)}}\|f\et-f\|_{1,P_{W}}&=\inf_{f\in\overline\cF_{(L,p,K)}}\int_{\sW}|f\et(w)-f(w)|dP_{W}(w)\nonumber\\
&\leq 19B(\lfloor\alpha\rfloor+1)^{2}d^{\lfloor\alpha\rfloor+(\alpha\vee1)/2}(NM)^{-2\alpha/d}.
\end{align}
On the other hand, given the equalities \eref{p-value} and \eref{L-value},
we have $p\geq342$ and $L\geq65$, for any $\alpha\in\R_{+}^{*}$. By applying Proposition~\ref{vc-dim}, we can derive through a basic computation that
\begin{align}\label{risk-vc-bound}
V(\overline\cF_{(L,p,K)})&\leq(L+1)\left(s+1\right)\log_{2}\cro{2\left(2e(L+1)\left(\frac{pL}{2}+1\right)\right)^{2}}\nonumber\\
&\leq C_{d}p^{2}L^{2}\log_{2}\left(pL^{2}\right)\nonumber\\
&\leq C_{\alpha,d}(NM)^{2}\cro{\log_{2}(2N)\log_{2}(2M)}^{3},
\end{align}
where $C_{d}$ only depends on $d$ and $C_{\alpha,d}$ only depends on $d$ and $\alpha$. Plugging \eref{risk-appro-bound} and \eref{risk-vc-bound} into \eref{risk-raw}, we can conclude that
\[
\E\cro{\ell(Q_{f\et},Q_{\widehat f})}\leq C_{\epsilon,\sigma,\alpha,d,B}\cro{(NM)^{-2\alpha/d}+\frac{NM}{\sqrt{n}}\left(\log_{2}(2N)\log_{2}(2M)\right)^{3/2}},
\]
where $C_{\epsilon,\sigma,\alpha,d,B}>0$ only depends on $\epsilon,\sigma,\alpha,d$ and $B$.
\end{proof}

\subsection{Proof of Proposition~\ref{approximate-add}}\label{approximation-composite-proof}
\begin{proof}
Prior to proving Proposition~\ref{approximate-add}, we will first introduce the following rules for network combination, which are extensively detailed in Section 7.1 of \cite{Schmidt2020}.

{\em Composition}: Let $f_{1}\in\overline\cF(L,{\bs{p}})$ and $f_{2}\in\overline\cF(L',{\bs{p}}')$ be such that $p_{L+1}=p_{0}'$. Let ${\bs{v}}\in\R^{p_{L+1}}$ be a vector. We define the composed network $f_{2}\circ\sigma_{\bs{v}}(f_{1})$, where $$\sigma_{\bs{v}}\begin{pmatrix}
y_{1}\\
\vdots\\
y_{p_{L+1}}
\end{pmatrix}=\begin{pmatrix}
\sigma(y_{1}-v_{1})\\
\vdots\\
\sigma(y_{p_{L+1}}-v_{p_{L+1}})
\end{pmatrix},$$ for any vector ${\bs{y}}=(y_{1},\ldots,y_{p_{L+1}})^{\top}\in\R^{p_{L+1}}$. Then $f_{2}\circ\sigma_{\bs{v}}(f_{1})$ belongs to the space $\overline\cF(L+L'+1,({\bs{p}},p'_{1},\ldots,p'_{L+1}))$.

{\em Parallelization}: Let $f_{1}$ and $f_{2}$ be two networks with an equal number of hidden layers and identical input dimensions. Specifically, let $f_{1}\in\overline\cF(L,{\bs{p}})$ and $f_{2}\in\overline\cF(L,{\bs{p}}')$, where $p_{0}=p'_{0}$. The parallelized network $(f_{1},f_{2})$ concurrently computes $f_{1}$ and $f_{2}$ within a joint network belonging to the class $\overline\cF(L,(p_{0},p_{1}+p'_{1},\ldots,p_{L+1}+p'_{L+1}))$.

We will also use the following inequality later in the proof. It can be derived through a minor modification of the proof of Lemma~3 in \cite{Schmidt2020}. 
\begin{lem}\label{composite-holder-sup}
Let $k\in\N^{*}$, ${\bf{d}}=(d_{0},\ldots,d_{k})\in(\N^{*})^{k+1}$, ${\bf{t}}=(t_{0},\ldots,t_{k})\in(\N^{*})^{k+1}$ with $t_{i}\leq d_{i}$ and ${\bs{\alpha}}=(\alpha_{0},\ldots,\alpha_{k})\in(\R_{+}^{*})^{k+1}$. For any $i\in\{0,\ldots,k\}$ and $j\in\{1,\ldots,d_{i+1}\}$ with $d_{k+1}=1$, let $h_{ij}\in\cH^{\alpha_{i}}(\cro{0,1}^{t_{i}},Q_{i})$ taking values in $\cro{0,1}$ for some $Q_{i}\geq1$ and $h_{i}=(h_{i1},\ldots,h_{id_{i+1}})^{\top}$. Then for any function $\widetilde h_{i}=(\widetilde h_{i1},\ldots,\widetilde h_{id_{i+1}})^{\top}$ with $\widetilde h_{ij}:\cro{0,1}^{t_{i}}\rightarrow\cro{0,1}$,
\begin{equation*}
\|h_{k}\circ\cdots\circ h_{0}-\widetilde h_{k}\circ\cdots\circ\widetilde h_{0}\|_{\infty}\leq \left(\prod_{i=0}^{k}Q_{i}\sqrt{d_{i}}\right)\sum_{i=0}^{k}|||h_{i}-\widetilde h_{i}|||_{\infty}^{\prod_{l=i+1}^{k}(\alpha_{l}\wedge1)},
\end{equation*}
where $|||f|||_{\infty}$ denotes the sup-norm of the function ${\bs{x}}\mapsto |f({\bs{x}})|_{\infty}$. 
\end{lem}

The essential strategy for establishing Proposition~\ref{approximate-add} is derived from a section of the proof of Theorem~1 in \cite{Schmidt2020}. However, we employ distinct fundamental networks as suggested by Proposition~\ref{appro-holder} to approximate functions with H\"older smoothness. This, in turn, leads to more specific neural network structures for approximating $f\et=f_{k}\circ\cdots\circ f_{0}$ compared to the sparsity-based networks considered in Theorem~1 of \cite{Schmidt2020}.

To begin with, we rewrite $$f\et=f_{k}\circ\cdots\circ f_{0}=g_{k}\circ\cdots\circ g_{0},$$ where $$g_{0}:=\frac{f_{0}}{2B}+\frac{1}{2},\quad g_{k}:=f_{k}(2B\cdot-B)$$ and $$g_{i}:=\frac{f_{i}(2B\cdot-B)}{2B}+\frac{1}{2}\quad\mbox{for all\ }i\in\{1,\ldots,k-1\}.$$
Given the condition $B\geq1$, we can readily confirm that
$g_{0j}\in\cH^{\alpha_{0}}(\cro{0,1}^{t_{0}},Q_{0})$, $g_{ij}\in\cH^{\alpha_{i}}(\cro{0,1}^{t_{i}},Q_{i})$, for $i\in\{1,\ldots,k-1\}$ and $g_{kj}\in\cH^{\alpha_{k}}(\cro{0,1}^{t_{k}},Q_{k})$, with $Q_{0}=1$, $Q_{i}=(2B)^{\alpha_{i}}$, for $i\in\{1,\ldots,k-1\}$ and $Q_{k}=2^{\alpha_{k}}B^{\alpha_{k}+1}$. 

We apply Proposition~\ref{appro-holder} to approximate each function $g_{ij}$, for all $j\in\{1,\ldots,d_{i+1}\}$, $i\in\{0,\ldots,k\}$. In particular, for all the functions $g_{i1},\ldots,g_{id_{i+1}}$, we take $N_{i}=1$, $M_{i}=\lceil n^{t_{i}/2(t_{i}+2\alpha_{i}^{*})}\rceil$ and consider a ReLU network $\overline\cF_{(L_{i},(t_{i},p_{i},\ldots,p_{i},1))}$ with $$p_{i}=114(\lfloor\alpha_{i}\rfloor+1)^{2}3^{t_{i}}t_{i}^{\lfloor\alpha_{i}\rfloor+1},$$
$$L_{i}=21(\lfloor\alpha_{i}\rfloor+1)^{2}\lceil n^{t_{i}/2(t_{i}+2\alpha_{i}^{*})}\rceil\lceil\log_{2}(8\lceil n^{t_{i}/2(t_{i}+2\alpha_{i}^{*})}\rceil)\rceil+2t_{i}.$$ According to Proposition~\ref{appro-holder}, there exists a function $\overline g_{ij}\in\overline\cF_{(L_{i},(t_{i},p_{i},\ldots,p_{i},1))}$ such that
\begin{align}\label{primary-appro}
\|\overline g_{ij}-g_{ij}\|_{\infty}&\leq19Q_{i}(\lfloor\alpha_{i}\rfloor+1)^{2}t_{i}^{\lfloor\alpha_{i}\rfloor+(\alpha_{i}\vee1)/2}(N_{i}M_{i})^{-2\alpha_{i}/t_{i}},\nonumber\\
&\leq19Q_{i}(\lfloor\alpha_{i}\rfloor+1)^{2}t_{i}^{\lfloor\alpha_{i}\rfloor+(\alpha_{i}\vee1)/2}(\lceil n^{t_{i}/2(t_{i}+2\alpha_{i}^{*})}\rceil)^{-2\alpha_{i}/t_{i}}.
\end{align}
Let $\widetilde g_{ij}=(\overline g_{ij}\vee0)\wedge1=1-(1-\overline g_{ij})_{+}$. It is straightforward to observe that $\widetilde g_{ij}$ assumes values in the interval $\cro{0,1}$. Recall that since $\overline g_{ij}\in\overline\cF_{(L_{i},(t_{i},p_{i},\ldots,p_{i},1))}$, it can be written as $$\overline g_{ij}=\overline M^{(i)}_{L_{i}}\circ\sigma\circ\cdots\circ\sigma\circ\overline M^{(i)}_{0},$$ for some linear transformations $\overline M^{(i)}_{0},\ldots,\overline M^{(i)}_{L_{i}}$. Let $M_{L_{i}+2}(x)=M_{L_{i}+1}(x)=1-x$, for any $x\in\R$. Then we have 
\begin{align*}
\widetilde g_{ij}&=M_{L_{i}+2}\circ\sigma\circ M_{L_{i}+1}\overline g_{ij}\\
&=M_{L_{i}+2}\circ\sigma\circ M_{L_{i}+1}\overline M^{(i)}_{L_{i}}\circ\sigma\circ\cdots\circ\sigma\circ\overline M^{(i)}_{0}\\
&=M_{L_{i}+2}\circ\sigma\circ \widetilde M^{(i)}_{L_{i}}\circ\sigma\circ\cdots\circ\sigma\circ\overline M^{(i)}_{0}
\end{align*}
where $\widetilde M^{(i)}_{L_{i}}=M_{L_{i}+1}\circ\overline M^{(i)}_{L_{i}}$. Hence, we deduce that $\widetilde g_{ij}\in\overline\cF_{(L_{i}+1,(t_{i},p_{i},\ldots,p_{i},1,1))}$. Furthermore, as each function $g_{ij}$ assumes values in the interval $\cro{0,1}$ due to the transformation, this implies that
\begin{equation}\label{truncated-inequa}
\|\sigma(\widetilde g_{ij})-g_{ij}\|_{\infty}=\|\widetilde g_{ij}-g_{ij}\|_{\infty}\leq\|\overline g_{ij}-g_{ij}\|_{\infty}.
\end{equation}

Next, we amalgamate these individual small networks by employing the fundamental operations of neural networks introduced at the outset of this proof. Note that $\overline\cF_{(L_{i}+1,(t_{i},p_{i},\ldots,p_{i},1,1))}\subset\overline\cF_{(L_{i}+1,(d_{i},p_{i},\ldots,p_{i},1,1))}$, for $t_{i}\leq d_{i}$. By the parallelization rule, the function $\widetilde g_{i}=(\widetilde g_{i1},\ldots, \widetilde g_{id_{i+1}})$ can be implemented by the ReLU neural network $\overline\cF_{(L_{i}+1,(d_{i},d_{i+1}p_{i},\ldots,d_{i+1}p_{i},d_{i+1},d_{i+1}))}$. A similar analysis implies that $\overline g_{k}$ can be implemented by the ReLU neural network $\overline\cF_{(L_{k},(d_{k},d_{k+1}p_{k},\ldots,d_{k+1}p_{k},d_{k+1}))}$. To construct the function $\widetilde{f}=\overline g_{k}\circ\widetilde g_{k-1}\cdots\circ\widetilde g_{0}$ that approximates the function $f\et=g_{k}\circ\cdots\circ g_{0}$, we apply the composition rule to amalgamate the networks we have considered earlier. It can be shown with a similar argument as we did before that for any $k\in\N^{*}$, $\widetilde g_{k-1}\circ\cdots\circ\widetilde g_{0}$ can be implemented by the ReLU neural network $$\overline\cF\left(\sum_{i=0}^{k-1}(L_{i}+1),\left(d_{0},\underbrace{d_{1}p_{0},\ldots,d_{1}p_{0}}_{L_{0}\ \mbox{times}},\ldots,d_{k-1},\underbrace{d_{k}p_{k-1},\ldots,d_{k}p_{k-1}}_{L_{k-1}\ \mbox{times}},d_{k},d_{k}\right)\right).$$ Note that $d_{i}\geq1$ for all $0\leq i\leq k+1$ and $p_{i}\geq 1$ for $0\leq i\leq k$. Denote $$\overline p=\max_{i=0,\ldots,k}d_{i+1}p_{i}.$$ Finally, we can conclude that the function $\widetilde{f}=\overline g_{k}\circ\widetilde g_{k-1}\cdots\circ\widetilde g_{0}$ can be implemented by the ReLU neural network $\overline\cF_{(\overline L,(d_{0},\overline p,\ldots,\overline p,d_{k+1}))}$ with $\overline L=k+\sum_{i=0}^{k}L_{i}$.

Recall that $Q_{0}=1$, $Q_{i}=(2B)^{\alpha_{i}}$, for $i\in\{1,\ldots,k-1\}$ and $Q_{k}=2^{\alpha_{k}}B^{\alpha_{k}+1}$. Combining Lemma~\ref{composite-holder-sup} with \eref{primary-appro} and \eref{truncated-inequa} yields the following upper bound for the approximation error,
\begin{align*}
&\inf_{f\in\overline\cF_{(\overline L,(d_{0},\overline p,\ldots,\overline p,d_{k+1}))}}\|f\et-f\|_{\infty}\\
&\leq (2B)^{1+\sum_{i=1}^{k}\alpha_{i}}\left(\prod_{i=0}^{k}\sqrt{d_{i}}\right)\cro{\sum_{i=0}^{k}C_{\alpha_{i},t_{i},B}^{\prod_{l=i+1}^{k}(\alpha_{l}\wedge1)}(\lceil n^{t_{i}/2(t_{i}+2\alpha_{i}^{*})}\rceil)^{-2\alpha^{*}_{i}/t_{i}}},
\end{align*}
where $$C_{\alpha_{i},t_{i},B}=19(2B)^{\alpha_{i}+1}(\lfloor\alpha_{i}\rfloor+1)^{2}t_{i}^{\lfloor\alpha_{i}\rfloor+(\alpha_{i}\vee1)/2}.$$ 
\end{proof}

\subsection{Proof of Corollary~\ref{composite-holder-riskbound}}\label{proof-composite-holder}
\begin{proof}
Firstly, we establish an upper bound for the VC-dimension of the ReLU neural network $\overline\cF_{(\overline L,\overline p,K)}$. Using the fact that for any $k$, ${\bf{d}}$, ${\bf{t}}$, ${\bs{\alpha}}$ and any $n\geq1$, $\overline L\geq L_{0}\geq65$, and $\overline p\geq p_{0}\geq342$, we can deduce, through the application of Proposition~\ref{vc-dim}, that
\begin{align}
V(\overline\cF_{(\overline L,\overline p,K)})&\leq C_{d_{0}}\overline p^{2}\overline L^{2}\log_{2}\left(\overline p\overline L^{2}\right)\nonumber\\
&\leq C_{k,{\bf{d}},{\bf{t}},{\bs{\alpha}}}\left(\sum_{i=0}^{k}L_{i}\right)^{2}\log_{2}\left(\sum_{i=0}^{k}L_{i}\right),\label{vc-add-explicit}
\end{align}
where $C_{k,{\bf{d}},{\bf{t}},{\bs{\alpha}}}>0$ is a numerical constant depending only on $k,{\bf{d}},{\bf{t}}$ and ${\bs{\alpha}}$. Combining \eref{riskbound} with the inequality \eref{tvconnect}, we obtain that for any $n\geq V(\overline\cF_{(\overline L,\overline p)})\geq V(\overline\cF_{(\overline L,\overline p,K)})$,
\begin{align}
\E\cro{\ell(Q_{f\et},Q_{\widehat f})}&\leq C_{\epsilon,\sigma}\cro{\inf_{f\in\cF_{(\overline L,\overline p,K)}}\|f\et-f\|_{1,P_{W}}+\sqrt{\frac{V(\overline\cF_{(\overline L,\overline p,K)})}{n}}}\nonumber\\
&\leq C_{\epsilon,\sigma}\cro{\inf_{f\in\overline\cF_{(\overline L,\overline p,K)}}\|f\et-f\|_{1,P_{W}}+\sqrt{\frac{V(\overline\cF_{(\overline L,\overline p,K)})}{n}}},\label{half-risk}
\end{align}
where the second inequality rises from the fact that $\cF_{(\overline L,\overline p,K)}$ is dense in $\overline\cF_{(\overline L,\overline p,K)}$ with respect to the sup-norm according to Lemma~\ref{network-dense}.

Finally, plugging the result provided in Proposition~\ref{approximate-add} and \eref{vc-add-explicit} into \eref{half-risk}, we can conclude that
\begin{align*}
\E\cro{\ell(Q_{f\et},Q_{\widehat f})}&\leq C_{\epsilon,\sigma,k,{\bf{d}},{\bf{t}},{\bs{\alpha}},B}\cro{\left(\sum_{i=0}^{k}n^{-\frac{\alpha_{i}^{*}}{t_{i}+2\alpha_{i}^{*}}}\right)+\left(\sum_{i=0}^{k}L_{i}\right)\sqrt{\frac{\log_{2}\left(\sum_{i=0}^{k}L_{i}\right)}{n}}}\nonumber\\
&\leq C_{\epsilon,\sigma,k,{\bf{d}},{\bf{t}},{\bs{\alpha}},B}\cro{\sum_{i=0}^{k}\left(n^{-\frac{\alpha_{i}^{*}}{t_{i}+2\alpha_{i}^{*}}}+\frac{L_{i}}{\sqrt{n}}\right)}\sqrt{\log_{2}\left(\sum_{i=0}^{k}L_{i}\right)}\nonumber\\
&\leq C_{\epsilon,\sigma,k,{\bf{d}},{\bf{t}},{\bs{\alpha}},B}\left(\sum_{i=0}^{k}n^{-\frac{\alpha_{i}^{*}}{t_{i}+2\alpha_{i}^{*}}}\right)(\log n)^{3/2}.
\end{align*}

\end{proof}

\subsection{Proof of Theorem~\ref{lower-composite}}\label{proof-lower-bound}
To establish lower bounds, we initially prove the following variant of Assouad's lemma.
\begin{lem}\label{assouad-lemma}
Let $\cP$ be a family of probabilities on a measurable space $(\sX,\cX)$. If for some integer $D\geq1$, there is a subset of $\cP$ of the form $\left\{P_{{\bs{\varepsilon}}},\;{\bs{\varepsilon}}\in\{0,1\}^{D}\right\}$ satisfying
\begin{listi}
\item there exists $\eta>0$ such that for all ${\bs{\varepsilon}},{\bs{\varepsilon}}'\in\{0,1\}^{D}$, $$\|P_{{\bs{\varepsilon}}}-P_{{\bs{\varepsilon}}'}\|_{TV}\geq\eta \delta({\bs{\varepsilon}},{\bs{\varepsilon}}')\quad \text{with}\quad \delta({\bs{\varepsilon}},{\bs{\varepsilon}}')=\sum_{j=1}^{D}\1_{\eps_{j}\ne \eps_{j}'}$$
\item there exists a constant $a\in\cro{0,1/2}$ such that 
\[
h^{2}\pa{P_{{\bs{\varepsilon}}},P_{{\bs{\varepsilon}}'}}\leq\frac{a}{n}\quad \text{for all ${\bs{\varepsilon}},{\bs{\varepsilon}}'\in \{0,1\}^{D}$ satisfying $\delta({\bs{\varepsilon}},{\bs{\varepsilon}}')=1$.}
\]
\end{listi}
Then for all measurable mappings $\widehat P:\sX^{n}\to \cP$,
%
\begin{equation}\label{concAssouad}
\sup_{P\in \cP}\E_{\gP}\cro{\|P-\widehat P(\bsX)\|_{TV}}\geq \frac{\eta D}{4}\max\ac{1-\sqrt{2a},\;\frac{1}{2}\left(1-\frac{a}{n}\right)^{2n}},
\end{equation}
where $\E_{\gP}$ denotes the expectation with respect to a random variable $\bsX=(\etc{X})$ with distribution $\gP=P\on$.
\end{lem}

\begin{proof}
Let ${\bs{\overline \eps}}$ minimize ${\bs{\eps}}\mapsto \|P-P_{{\bs{\eps}}}\|_{TV}$ over $\{0,1\}^{D}$ for a given probability $P$ on $(\sX,\cX)$. Note that for all ${\bs{\eps}}\in\{0,1\}^{D}$,
\[
\|P_{{\bs{\eps}}}-P_{\overline{\bs{\eps}}}\|_{TV}\leq\|P-P_{{\bs{\eps}}}\|_{TV}+\|P-P_{\overline{\bs{\eps}}}\|_{TV}\leq 2\|P-P_{{\bs{\eps}}}\|_{TV}.
\]
Thus, using property (i), we have for all ${\bs{\eps}}\in\{0,1\}^{D}$:
\[
\|P_{{\bs{\eps}}}-P\|_{TV}\geq \frac{\eta}{2}\delta({\bs{\eps}},\overline{\bs{\eps}})=\sum_{i=1}^{D}\cro{\eps_{i}\ell_{i}(P)+(1-\eps_{i})\ell_{i}'(P)},
\]
where $\ell_{i}(P)=(\eta/2)\1_{\overline \eps_{i}=0}$ and $\ell_{i}'(P)=(\eta/2)\1_{\overline \eps_{i}=1}$, for  $i\in\{1,\ldots,D\}$. Finally, the conclusion follows by applying a version of Assouad's lemma from \cite{Birge1986} with $\beta_{i}=a/n$, for all $i\in\{1,\ldots,D\}$ and $\alpha=\eta/2$.
\end{proof}

Now we prove Theorem~\ref{lower-composite}. The roadmap is to first find a suitable collection of probabilities $\cP$ then apply Lemma~\ref{assouad-lemma} to derive the lower bound. 

The construction idea is inspired by the proof of Theorem~3 of \cite{Schmidt2020}. Denote $i^{*}\in\argmin_{i=0,\ldots,k}\alpha^{*}_{i}/(2\alpha^{*}_{i}+t_{i})$. For simplicity, we write $t^{*}=t_{i^{*}}$, $\alpha^{*}=\alpha_{i^{*}}$ and $\alpha^{**}=\alpha^{*}_{i^{*}}$. We define $N_{n}=\lfloor\rho n^{1/(2\alpha^{**}+t^{*})}\rfloor$, $h_{n}=1/N_{n}$ and $\Lambda=\left\{0,h_{n},\ldots,(N_{n}-1)h_{n}\right\}$.

Following the construction outlined on page 93 of \cite{Tsybakov2009}, we consider the function $$\cK(x)=a\exp\left(-\frac{1}{1-(2x-1)^{2}}\right)\1_{|2x-1|\leq1}$$ with $a>0$. Provided that $a$ is sufficiently small, we have $\cK\in\cH^{\alpha^{*}}(\R,1)$ with support on $\cro{0,1}$.
Moreover, for any $\beta\in\N$ satisfying $\beta\leq\lfloor\alpha^{*}\rfloor$, the $\beta$-th derivative of $\cK$ is zero at both $x=0$ and $x=1$, i.e., $\cK^{(\beta)}(0)=\cK^{(\beta)}(1)=0$.
We define
the function $\psi_{\bf{u}}$ on $\cro{0,1}^{t^{*}}$ as 
\[
\psi_{\bf{u}}(w_{1},\ldots,w_{t^{*}})=h_{n}^{\alpha^{*}}\prod_{j=1}^{t^{*}}\cK\left(\frac{w_{j}-u_{j}}{h_{n}}\right),
\]
where ${\bf{u}}=(u_{1},\ldots,u_{t^{*}})\in\cU_{n}=\left\{(u_{1},\ldots,u_{t^{*}}),\;u_{i}\in\Lambda\right\}$. Note that for any ${\bf{u}}, {\bf{u}}'\in\cU_{n}$, ${\bf{u}}\not={\bf{u}}'$, the supports of $\psi_{\bf{u}}$ and $\psi_{{\bf{u}}'}$ are disjoint. For any ${\bs{\beta}}=(\beta_{1},\ldots,\beta_{t^{*}})\in\N^{t^{*}}$ satisfying $\sum_{j=1}^{t^{*}}\beta_{j}\leq\lfloor\alpha^{*}\rfloor$, it holds that $\|\partial^{\bs{\beta}}\psi_{\bf{u}}\|_{\infty}\leq1$ due to the fact that $\cK\in\cH^{\alpha^{*}}(\R,1)$. Set $\cI_{\bf{u}}=\cro{u_{1},u_{1}+h_{n}}\times\cdots\times\cro{u_{t^{*}},u_{t^{*}}+h_{n}}$. Moreover, for any ${\bs{\beta}}=(\beta_{1},\ldots,\beta_{t^{*}})$ with $\sum_{j=1}^{t^{*}}\beta_{j}=\lfloor\alpha^{*}\rfloor$, with the fact that $\cK\in\cH^{\alpha^{*}}(\R,1)$ and triangle inequality, we obtain that for any ${\bs{x}}, {\bs{y}}\in\cI_{\bf{u}}$,
\[
\frac{\big|\partial^{\bs{\beta}}\psi_{\bf{u}}({\bs{x}})-\partial^{\bs{\beta}}\psi_{\bf{u}}({\bs{y}})\big|}{\|{\bs{x}}-{\bs{y}}\|_{2}^{\alpha^{*}-\lfloor\alpha^{*}\rfloor}}\leq t^{*}.
\]
Therefore, we have $\psi_{\bf{u}}\in\cH^{\alpha^{*}}(\cI_{\bf{u}},t^{*})$.
For any vector ${\bs{\varepsilon}}=(\eps_{{\bf{u}}})_{{\bf{u}}\in\cU_{n}}\in\{0,1\}^{|\cU_{n}|}$, define the function $\phi_{{\bs{\varepsilon}}}$ on $\cro{0,1}^{t^{*}}$ as $$\phi_{{\bs{\varepsilon}}}(w_{1},\ldots,w_{t^{*}})=\sum_{{\bf{u}}\in\cU_{n}}\eps_{{\bf{u}}}\psi_{{\bf{u}}}(w_{1},\ldots,w_{t^{*}}).$$ Given that $\cK\in\cH^{\alpha^{*}}(\R,1)$ and $\cK^{(\beta)}(0)=\cK^{(\beta)}(1)=0$, for any $\beta\leq\lfloor\alpha^{*}\rfloor$, it is not difficult to verify that $\phi_{{\bs{\varepsilon}}}\in\cH^{\alpha^{*}}(\cro{0,1}^{t^{*}},2t^{*})$.

Let $d'_{i}=\min\{d_{0},\ldots,d_{i}\}$, for all $i\in\{0,\ldots,k\}$. For $0\leq i<i^{*}$, we denote $f_{i}({\bs{w}})=(w_{1},\ldots,w_{d_{i+1}})^{\top}$, if $d_{i+1}=d'_{i+1}$; otherwise, we set $f_{i}({\bs{w}})=(w_{1},\ldots,w_{d'_{i}},0,\ldots,0)^{\top}$. We denote $f_{{\bs{\varepsilon}},i^{*}}({\bs{w}})=(\phi_{{\bs{\varepsilon}}}(w_{1},\ldots,w_{t^{*}}),0,\ldots,0)^{\top}$, $f_{i}({\bs{w}})=(w_{1}^{\alpha_{i}\wedge1},0,\ldots,0)^{\top}$, for $i^{*}<i\leq k-1$ and $f_{k}({\bs{w}})=w_{1}^{\alpha_{k}\wedge1}$. Let $\cA=\prod_{l=i^{*}+1}^{k}(\alpha_{l}\wedge1)$. Since $t_{j}\leq\min(d_{0},\ldots,d_{j-1})$, we can set
\begin{align*}
f_{{\bs{\varepsilon}}}({\bs{w}})&=f_{k}\circ\cdots\circ f_{i^{*}+1}\circ f_{{\bs{\varepsilon}},i^{*}}\circ f_{i^{*}-1}\circ\cdots\circ f_{0}({\bs{w}})\\
&=\sum_{{\bf{u}}\in\cU_{n}}\eps_{{\bf{u}}}\cro{\psi_{\bf{u}}(w_{1},\ldots,w_{t^{*}})}^{\cA}.
\end{align*}
Consequently, we can observe that the resulting function $f_{{\bs{\varepsilon}}}$ belongs to the class $\cF(k,{\bf{d}},{\bf{t}},{\bm{\alpha}},B)$ when $B$ is sufficiently large.

Since $W$ is uniformly distributed on $\cro{0,1}^{d_{0}}$, we can compute
\begin{equation}\label{l2-hamming-connect}
\|f_{{\bs{\varepsilon}}}-f_{{\bs{\varepsilon}}'}\|_{2}^{2}=\delta({\bs{\varepsilon}},{\bs{\varepsilon}}')h_{n}^{2\alpha^{**}+t^{*}}\|\cK^{\cA}\|_{2}^{2t^{*}}
\end{equation}
and
\begin{equation}\label{l1-hamming-connect}
\|f_{{\bs{\varepsilon}}}-f_{{\bs{\varepsilon}}'}\|_{1}=\delta({\bs{\varepsilon}},{\bs{\varepsilon}}')h_{n}^{\alpha^{**}+t^{*}}\|\cK^{\cA}\|_{1}^{t^{*}},
\end{equation}
where $\delta(\cdot,\cdot)$ denotes the Hamming distance. For any $P_{f_{{\bs{\varepsilon}}}}=Q_{f_{{\bs{\varepsilon}}}}\cdot P_{W}$ and $P_{f_{{\bs{\varepsilon}}'}}=Q_{f_{{\bs{\varepsilon}}'}}\cdot P_{W}$, where $P_{W}$ is the uniform distribution on $\cro{0,1}^{d_{0}}$, we can derive that
\begin{align}
h^{2}(P_{f_{{\bs{\varepsilon}}}},P_{f_{{\bs{\varepsilon}}'}})&=\int_{\sW}\left(1-\exp\cro{-\frac{|f_{{\bs{\varepsilon}}}(w)-f_{{\bs{\varepsilon}}'}(w)|^{2}}{8\sigma^{2}}}\right)dP_{W}(w)\nonumber\\
&\leq\int_{\sW}\frac{|f_{{\bs{\varepsilon}}}(w)-f_{{\bs{\varepsilon}}'}(w)|^{2}}{8\sigma^{2}}dP_{W}(w)\nonumber\\
&=\frac{\|f_{{\bs{\varepsilon}}}-f_{{\bs{\varepsilon}}'}\|^{2}_{2}}{8\sigma^{2}}.\label{h-upper-b}
\end{align}
According to Lemma~\ref{l1-ell}, we can deduce that for $W$ uniformly distributed on $\cro{0,1}^{d_{0}}$,
\begin{equation}\label{tv-lower-b}
\ell(Q_{f_{{\bs{\varepsilon}}}},Q_{f_{{\bs{\varepsilon}}'}})=\|P_{f_{{\bs{\varepsilon}}}}-P_{f_{{\bs{\varepsilon}}'}}\|_{TV}\geq \frac{0.78}{\sqrt{2\pi}\sigma}\|f_{{\bs{\varepsilon}}}-f_{{\bs{\varepsilon}}'}\|_{1},
\end{equation}
provided $\rho\geq1+\cro{\|\cK^{\cA}\|_{1}^{t^{*}}/(\sqrt{2\pi}\sigma)}^{1/\alpha^{**}}$ such that $h_{n}^{\alpha^{**}}\leq\sqrt{2\pi}\sigma/\|\cK^{\cA}\|_{1}^{t^{*}}$. Putting \eref{l2-hamming-connect}, \eref{l1-hamming-connect}, \eref{h-upper-b} and \eref{tv-lower-b} together, we observe that the family of probabilities $\cP=\{P_{\bsg_{{\bs{\varepsilon}}}},\;{\bs{\varepsilon}}\in\{0,1\}^{|\cU_{n}|}\}$ satisfies the assumptions of Lemma~\ref{assouad-lemma} with $D=N_{n}^{t^{*}}$, $$\eta=\frac{0.78}{\sqrt{2\pi}\sigma}h_{n}^{\alpha^{**}+t^{*}}\|\cK^{\cA}\|_{1}^{t^{*}}\quad\text{and}\quad a=\frac{1}{8\sigma^{2}}nh_{n}^{2\alpha^{**}+t^{*}}\|\cK^{\cA}\|_{2}^{2t^{*}}.$$ Finally, taking the constant $$\rho\geq\cro{1+\left(\frac{\|\cK^{\cA}\|_{1}^{t^{*}}}{\sqrt{2\pi}\sigma}\right)^{\frac{1}{\alpha^{**}}}}\vee\cro{1+\left(\frac{\|\cK^{\cA}\|_{2}^{2t^{*}}}{\sigma^{2}}\right)^{\frac{1}{2\alpha^{**}+t^{*}}}}$$ such that $h_{n}^{\alpha^{**}}\leq(n\|\cK^{\cA}\|_{2}^{2t^{*}}/\sigma^{2})^{-\frac{\alpha^{**}}{2\alpha^{**}+t^{*}}}\wedge\left(\sqrt{2\pi}\sigma/\|\cK^{\cA}\|_{1}^{t^{*}}\right)$, we derive by Lemma~\ref{assouad-lemma} that there exists some constant $c>0$ such that 
\[
\inf_{\widehat f}\sup_{f\et\in\cF(k,{\bf{d}},{\bf{t}},{\bm{\alpha}},B)}\E\cro{\ell(Q_{f\et},Q_{\widehat f})}\geq cn^{-\frac{\alpha^{**}}{2\alpha^{**}+t^{*}}}.
\]

\bibliographystyle{apalike}

\begin{thebibliography}{}
\bibitem[Baraud, 2016]{Baraud2016}
Baraud, Y. (2016). 
\newblock Bounding the expectation of the supremum of an empirical process over a (weak) {VC}-major class.
\newblock {\em Electron. J. Statist.}, {\bf10}, 1709--1728.

\bibitem[Baraud, 2021]{Baraud2021}
Baraud, Y. (2021).
\newblock Tests and estimation strategies associated to some loss functions.
\newblock {\em Probab. Theory Related Fields}, {\bf180}, 799--846.

\bibitem[Baraud and Birg{\'e}, 2018]{Baraud2018}
Baraud, Y. and Birg{\'e}, L. (2018).
\newblock Rho-estimators revisited: {G}eneral theory and applications.
\newblock {\em Ann. Statist.}, {\bf46}, 3767--3804.

\bibitem[Baraud and Birg\'e, 2014]{Baraud:2011fk}
Baraud, Y. and Birg{\'e}, L. (2014). 
\newblock Estimating composite functions by model selection.
\newblock {\em Ann. Inst. H. Poincar{\'e} Probab. Statist.},
{\bf50}, 285--314.

\bibitem[Baraud et~al., 2017]{MR3595933}
Baraud, Y., Birg{\'e}, L. and Sart, M. (2017).
\newblock A new method for estimation and model selection: {$\rho$}-estimation.
\newblock {\em Invent. Math.}, {\bf207}, 425--517.

\bibitem[Baraud and Chen, 2020]{Baraud2020}
Baraud, Y. and Chen, J. (2020). 
\newblock Robust estimation of a regression function in exponential families.
\newblock {\em arXiv}:2011.01657.

\bibitem[Baraud et~al., 2022]{Baraud2022}
Baraud, Y., Halconruy, H. and Maillard, G. (2022). 
\newblock Robust density estimation with the $\L_{1}$-loss. Applications to the estimation of a density on the line satisfying a shape constraint.
\newblock {\em arXiv}:2205.10524.

\bibitem[Bartlett et~al., 2019]{BarlettJMLR}
Bartlett, P. L., Harvey, N., Liaw, C. and Mehrabian, A. (2019).
\newblock Nearly-tight {VC}-dimension and pseudodimension bounds for piecewise linear neural networks.
\newblock {\em J. Mach. Learn. Res.}, {\bf20}, 1--17.

\bibitem[Barron, 2019]{Barron2019}
Barron, J. T. (2019).
\newblock A general and adaptive robust loss function.
In \textit{Proceedings of the IEEE/CVF Conference on Computer Vision and Pattern Recognition}, 4331--4339.


\bibitem[Bassett and Koenker, 1978]{Bassett1978}
Bassett, G. and Koenker, M. (1992).
\newblock Asymptotic theory of least absolute error regression.
\newblock {\em J. Amer. Statist. Assoc.}, {\bf73}, 618--622.


\bibitem[Bauer and Kohler, 2019]{Kohler2019}
Bauer, B. and Kohler, M. (2019).
\newblock On deep learning as a remedy for the curse of dimensionality in nonparametric regression.
\newblock {\em Ann. Statist.}, {\bf47}, 2261--2285.

\bibitem[Beaton and Tukey, 1974]{Beaton1974}
Beaton, A. E. and Tukey, J. W. (1974).
\newblock The fitting of power series, meaning polynomials, illustrated on band-spectroscopic data.
\newblock {\em Technometrics}, {\bf16}, 147--185.


\bibitem[Birg{\'e}, 1986]{Birge1986}
Birg{\'e}, L. (1986).
\newblock On estimating a density using {H}ellinger distance and some other
  strange facts.
\newblock {\em Probab. Theory Relat. Fields}, {\bf71}, 271--291.

\bibitem[Chen, 2022]{ChenModSel}
Chen, J. (2022).
\newblock Estimating a regression function in exponential families by model selection.
\newblock {\em arXiv}:2203.06656.

\bibitem[Chen et~al., 2019]{Minshuo2019}
Chen, M., Jiang, H., Liao, W. and Zhao, T. (2019).
\newblock Efficient approximation of deep {R}e{LU} networks for functions on low dimensional manifolds.
In \textit{32th Advances in Neural Information Processing Systems}, NeurIPS.

\bibitem[Chen et~al., 2022]{Minshuo2022}
Chen, M., Jiang, H., Liao, W. and Zhao, T. (2022).
\newblock Nonparametric regression on low-dimensional manifolds using deep ReLU networks: function approximation and statistical recovery.
\newblock {\em Inf. Inference}, {\bf11}, 1203–1253.



\bibitem[Donoho and Johnstone, 1998]{1024691081}
Donoho, D. L. and Johnstone, I. M. (1998).
\newblock Minimax estimation via wavelet shrinkage.
\newblock {\em Ann. Statist.}, {\bf26}, 879--921.

\bibitem[Donoho et~al., 1995]{2345967}
Donoho, D. L., Johnstone, I. M., Kerkyacharian, G. and Picard, D. (1995).
\newblock Wavelet shrinkage: {A}symptopia?
\newblock {\em J. Roy. Statist. Soc., Ser. B}, {\bf57}, 301--369.

\bibitem[Fan, 1992]{Fan1992}
Fan, J. (1992).
\newblock Design-adaptive nonparametric regression.
\newblock {\em J. Amer. Statist. Assoc.}, {\bf87}, 998--1004.

\bibitem[Fan, 1993]{Fan1993}
Fan, J. (1993).
\newblock Local linear regression smoothers and their minimax efficiencies.
\newblock {\em Ann. Statist.}, {\bf21}, 196--216.

\bibitem[Friedman, 1991]{Friedman1991}
Friedman, J. (1991).
\newblock Multivariate adaptive regression splines.
\newblock {\em Ann. Statist.}, {\bf19}, 1--67.

\bibitem[Gy\"orfi et~al., 2002]{Gyorfi2002}
Gy\"orfi, L., Kohler, M., Krzyzak, A. and Walk, H. (2002).
\newblock {\em A Distribution-Free Theory of Nonparametric Regression}.
\newblock Springer Series in Statistics. Springer-Verlag, New York.

\bibitem[Horowitz and Mammen, 2007]{Mammen2007}
Horowitz, J. L. and Mammen, E. (2007).
\newblock Rate-optimal estimation for a general class of nonparametric regression models with unknown link functions.
\newblock {\em Ann. Statist.}, {\bf35}, 2589--2619.

\bibitem[Huber, 1973]{Huber1973}
Huber, P. J. (1973).
\newblock Robust regression: asymptotics, conjectures and {M}onte {C}arlo.
\newblock {\em Ann. Statist.},
{\bf1}, 799--821.

\bibitem[Jiao et~al., 2021]{Jiao2021}
Jiao, Y., Shen, G., Lin, Y. and Huang, J. (2021). 
\newblock Robust nonparametric regression with deep neural networks.
\newblock {\em arXiv}:2107.10343.

\bibitem[Jiao et~al., 2023]{Jiao2023}
Jiao, Y., Shen, G., Lin, Y. and Huang, J. (2023). 
\newblock Deep nonparametric regression on approximate manifolds: non-asymptotic error bounds with polynomial prefactors.
\newblock {\em Ann. Statist.}, {\bf51}, 691--716.

\bibitem[Kohler and Langer, 2021]{Sophie2021}
Kohler, M. and Langer, S. (2021). 
\newblock On the rate of convergence of fully connected deep neural network regression estimates.
\newblock {\em Ann. Statist.}, {\bf49}, 2231--2249.

\bibitem[Lederer, 2020]{Lederer2020}
Lederer, J. (2020). 
\newblock Risk bounds for robust deep learning.
\newblock {\em arXiv}:2009.06202.


\bibitem[Nadaraya, 1964]{Nadaraya1964}
Nadaraya, E. A. (1964)
\newblock On estimating regression.
\newblock {\em Theory Probab. its Appl.}, {\bf9}, 141--142.

\bibitem[Nakada and Imaizumi, 2020]{Nakada2020}
Nakada, R. and Imaizumi, M. (2020)
\newblock Adaptive approximation and generalization of deep neural network with intrinsic dimensionality.
\newblock {\em J. Mach. Learn. Res.}, {\bf21}, 1--38.

\bibitem[Padilla et al., 2022]{Padilla2022}
Padilla, O. H. M., Tansey, W. and Chen, Y. (2022)
\newblock Quantile regression with {R}e{LU} networks: {E}stimators and minimax rates.
\newblock {\em J. Mach. Learn. Res.}, {\bf23}, 1--42.

\bibitem[Schmidt-Hieber, 2019]{SchmidtMani}
Schmidt-Hieber, J. (2019). 
\newblock Deep {R}e{LU} network approximation of functions on a manifold.
\newblock {\em arXiv}:1908.00695.

\bibitem[Schmidt-Hieber, 2020]{Schmidt2020}
Schmidt-Hieber, J. (2020). 
\newblock Nonparametric regression using deep neural networks with {R}e{LU} activation function.
\newblock {\em Ann. Statist.}, {\bf48}, 1875--1897.

\bibitem[Stone, 1982]{Stone1982}
Stone, C. J. (1982). 
\newblock Optimal global rates of convergence for nonparametric regression.
\newblock {\em Ann. Statist.}, {\bf10}, 1040--1053.

\bibitem[Stone, 1985]{Stone1985}
Stone, C. J. (1985). 
\newblock Additive regression and other nonparametric models.
\newblock {\em Ann. Statist.},
{\bf13}, 689--705.

\bibitem[Suzuki, 2019]{suzuki2019iclr}
Suzuki, T. (2019). 
\newblock Adaptivity of deep {R}e{LU} network for learning in {B}esov and mixed smooth {B}esov spaces: optimal rate and curse of dimensionality.
In \textit{7th International Conference on Learning Representations}, ICLR.

\bibitem[Suzuki and Nitanda, 2021]{suzuki2019deep}
Suzuki, T. and Nitanda, A. (2021).
\newblock Deep learning is adaptive to intrinsic dimensionality of model smoothness in anisotropic {B}esov space.
In \textit{34th Advances in Neural Information Processing Systems}, NeurIPS.

\bibitem[Tsybakov, 2009]{Tsybakov2009}
Tsybakov, A. B. (2009).
\newblock {\em Introduction to Nonparametric Estimation}.
\newblock Springer Series in Statistics. Springer-Verlag, New York.


\bibitem[van~der Vaart and Wellner, 1996]{MR1385671}
van~der Vaart, A.~W. and Wellner, J.~A. (1996).
\newblock {\em Weak Convergence and Empirical Processes. With Applications to
  Statistics}.
\newblock Springer Series in Statistics. Springer-Verlag, New York.

\bibitem[van~der Vaart and Wellner, 2009]{MR2797943}
van~der Vaart, A. and Wellner, J.~A. (2009).
\newblock A note on bounds for {VC} dimensions.
\newblock In {\em High Dimensional Probability {V}: the {L}uminy volume},
  volume~5 of {\em Inst. Math. Stat. Collect.}, pages 103--107. Inst. Math.
  Statist., Beachwood, OH.
  
\bibitem[Wahba, 1990]{Wahba1990}
Wahba, G. (1990).
\newblock {\em Spline Models for Observational Data}.
\newblock Society for Industrial and Applied Mathematics.

\bibitem[Watson, 1964]{Watson1964}
Watson, G. S. (1964).
\newblock Smooth regression analysis.
 \newblock {\em Sankhyā: Indian J. Stat., Ser. A}, {\bf26}, 359--372.
\end{thebibliography}

\end{document}